\documentclass[12pt]{amsart}
\usepackage[utf8]{inputenc}
\usepackage[T1]{fontenc}
\usepackage{lmodern}
\usepackage[english]{babel}
\usepackage{amsmath,mathtools,amsthm,amssymb,amsfonts}
\usepackage{color}
\usepackage{cancel}
\usepackage{mathrsfs}
\usepackage{graphicx}
\usepackage{enumitem}
\usepackage{tikz} 
\usepackage{xcolor}
\usepackage{pgfplots}
\usepackage{esvect}
\usepackage{pdfpages}
\usepackage{nicefrac}
\usepackage{cleveref}
\usepackage{float}
\usepackage{caption}
\usepackage{subcaption}
\usetikzlibrary{shapes,arrows}
\usetikzlibrary{calc}
\usetikzlibrary{shapes.geometric}
\usetikzlibrary{shapes.arrows}
\usetikzlibrary{arrows.meta}
\usetikzlibrary{decorations.markings}
\usetikzlibrary{fit}
\usetikzlibrary{patterns}
\usetikzlibrary{hobby}
\usepgfplotslibrary{patchplots}
\pgfplotsset{compat=1.11}

\newcommand{\definedas}{\mathrel{\raise.095ex\hbox{\rm :}\mkern-5.2mu=}}

\newcommand{\R}{\mathbb{R}}

\renewcommand{\d}{\,\mathrm{d}}

\newcommand{\btr}[1]{\left\vert#1\right\vert}

\newcommand{\Ric}{\mathrm{Ric}}

\newcommand{\grad}{\text{grad}}

\makeatletter
\newcommand*{\rom}[1]{\expandafter\@slowromancap\romannumeral #1@}
\makeatother

\theoremstyle{plain}
\newtheorem{thm}{Theorem}[section]
\newtheorem{prop}[thm]{Proposition}

\newtheorem{lem}[thm]{Lemma}

\theoremstyle{definition}
\newtheorem{defi}[thm]{Definition}

\newtheorem{bem}[thm]{Remark}
\newtheorem{kor}[thm]{Corollary}

\begin{document}

\begin{center}
\vspace{1.3cm}
{\LARGE \textsc{Some new perspectives on the Kruskal--Szekeres extension with applications to photon surfaces}\par}
\end{center}
\vspace{2cm}
\begin{center}
\textbf{Carla Cederbaum, Markus Wolff}
\end{center}
\begin{center}
\vspace{1cm}
\end{center}
\thispagestyle{empty}
\section*{Abstract}
It is a well-known fact that the Schwarzschild spacetime admits a maximal spacetime extension in null coordinates which extends the exterior Schwarzschild region past the Killing horizon, called the Kruskal--Szekeres extension. 
This method of extending the Schwarzschild spacetime was later generalized by Brill--Hayward to a class of spacetimes of ``profile $h$'' across non-degenerate Killing horizons. Circumventing analytical subtleties in their approach, we reconfirm this fact by reformulating the problem as an ODE, and showing that the ODE admits a solution if and only if the naturally arising Killing horizon is non-degenerate. Notably, this approach lends itself to discussing regularity across the horizon for non-smooth metrics.

We will discuss applications to the study of photon surfaces, extending results by Ceder\-baum--Galloway and Cederbaum--Jahns--Vi\v{c}\'{a}nek-Mart\'{i}nez beyond the Killing horizon. In particular, our analysis asserts that photon surfaces approaching the Killing horizon must necessarily cross it.\newpage

\section{Introduction}\label{sec_intro}
\addtocounter{section}{0}
\addtocounter{thm}{0}
	In 1960, Kruskal and Szekeres both independently found the same extension of the Schwarz\-schild spacetime, which we now know to be maximal \cite{sbierski} and which is called the Kruskal--Szekeres extension \cite{kruskal, szekeres}.
	This method has been adapted to many other spacetime geometries of the general form $\mathcal{T}\times\mathcal{N}$ with metric
	\begin{align}\label{intro_metricform}
		g=-h(r)\d t^2+\frac{1}{h(r)}\d r^2+r^2g^{\mathcal{N}},
	\end{align}
	with $\mathcal{T}=\R\times\mathcal{I}$ for some open interval $\mathcal{I}\subseteq(0,\infty)$ and a complete Riemannian manifold $(\mathcal{N},g^{\mathcal{N}})$,
	to study various topics of near horizon geometry (see e.g. Gibbons--Hawking \cite{gibbhawk}, Grayson--Brill \cite{gravbrill}, Qiu--Traschen \cite{qiutraschen}). A rather general method for extending this class of spacetimes was developed by Walker \cite{walker70} in 1970, assuming that the metric coefficient $h$ has a certain algebraic structure. The general case was covered by Brill--Hayward \cite{brillhayward} who realized that Kruskal-like coordinates can be constructed across any non-degenerate Killing horizon. Brill and Hayward generalize the construction by Kruskal and Szekeres by introducing a suitable \emph{tortoise function}, see Section \ref{sec_globaltortoisefunc} below. In a numerical approach for the construction of Penrose--Carter diagrams, Schindler--Aguirre \cite{Schindler2018AlgorithmsFT} computed a global such tortoise function as the limit of a complex path integral, assuming real analyticity of $h$. Here, instead of studying the local properties of a tortoise function near the Killing horizon, we construct the spacetime extension from the solution of a global ODE (Theorem \ref{thm_main1}). We recover the result of Brill and Hayward by showing that the ODE is uniquely solvable (up to scaling) across the Killing horizon if and only if the Killing horizon is non-degenerate (Proposition \ref{prop_central}). In particular, we can then a posteriori recover a tortoise function in the style of Brill and Hayward. As the solvability relies on a version of l'H\^opital's Rule, our construction allows for precise regularity statements across the horizon for non-smooth metrics (Theorem \ref{thm_main2}). The constructed spacetime extensions are $C^2$-inextendable and geodesically incomplete under suitable conditions on $h$ (Corollary \ref{kor_inextendability}). 
	
	The Kruskal--Szekeres extension and its generalizations have many interesting applications, see e.g. \cite{gravbrill, gibbhawk, roesch, Schindler2018AlgorithmsFT}. Whenever these applications remain meaningful for non-smooth metrics, one can generalize them by applying our technique. Here, we focus on the theoretical and numerical computation of Penrose--Carter diagrams by Schindler--Aguirre \cite{Schindler2018AlgorithmsFT}, see Section \ref{sec_globaltortoisefunc}. In addition, we utilize the constructed generalized Kruskal--Szekeres coordinates and extension to analyze the behaviour of symmetric photon surfaces asymptotically near null infinity, near a non-degenerate Killing horizon, and inside the black and white hole regions of the extended spacetimes. In particular, we assert that photon surfaces approaching a non-degenerate Killing horizon must cross it (Theorem \ref{prop_crossing}) and those asymptoting to an asymptotically flat infinity ($h\to 1$ as $r\to\infty$) in fact asymptote to a lightcone (Theorem \ref{prop_lightconeasymptotics}). This extends and complements the work of Cederbaum--Galloway \cite{cedgal} and Cederbaum--Jahns--Vi\v{c}\'{a}nek-Mart\'{i}nez \cite{cedoliviasophia}. \newline\\
	This paper is structured as follows:
	In Section \ref{sec_prelim}, we will introduce the notation used throughout the paper.
	In Section \ref{sec_extension}, we will reduce the construction of a generalized Kruskal--Szekeres extension to the existence of solutions of a suitable ODE, and show that the ODE admits a solution if and only if the Killing horizon is non-degenerate. We also briefly touch upon inextendability and geodesic incompleteness.
	In Section \ref{sec_globaltortoisefunc}, we will comment on the construction of a global tortoise function using the above ODE.
	In Section \ref{sec_photonsurf}, we give an application of the generalized Kruskal--Szekeres coordinates and extension by extending the work of Cederbaum--Galloway and Cederbaum--Jahns--Vi\v{c}\'{a}nek-Mart\'{i}nez on symmetric photon surfaces across the horizon.

\subsection*{Acknowledgements.}
The authors would like to thank Gregory J. Galloway, Klaus Kr\"oncke, and Charles Robson for helpful comments and questions and Axel Fehrenbach and Maria Zioga for help with the graphics.

The first named author would like to extend thanks to the Mittag-Leffler Institute for allowing her to work in stimulating environments. The work of the first named author is supported by the focus program on Geometry at Infinity (Deutsche Forschungsgemeinschaft,  SPP 2026) and by the Institutional Strategy of the University of T\"ubingen (Deutsche Forschungsgemeinschaft, ZUK 63).

\setcounter{section}{1}
\section{Preliminaries}\label{sec_prelim}
	We consider $(n+1)$-dimensional spacetimes of a certain class $\mathfrak{H}$, which carry metrics of the above form \eqref{intro_metricform} and are fully determined by a choice of a metric coefficient $h$ and of an $(n-1)$-dimensional Riemannian manifold $({\mathcal{N}},g_{\mathcal{N}})$. Here $h\colon(0,\infty)\to\R$ is smooth, unless otherwise stated, with positive, real zeros $r_0\definedas 0<r_1<\dotsc r_i<\infty=:r_{N+1}$, $N\ge 1$. Then we say that
	\begin{align*}
		M_i&=\R\times(r_{i-1},r_i)\times {\mathcal{N}},\\
		g&=-h(r)\d t^2+\frac{1}{h(r)}\d r^2+r^2g_{\mathcal{N}},
	\end{align*}
	is a \emph{spacetime of class $\mathfrak{H}$}, where $i\in\{1,\dotsc, N+1\}$. Wherever $h>0$, the metric $g$ is static with timelike Killing vector field $\partial_t$, however as we aim to look inside black hole horizons or past cosmological horizons, we also want to consider regions where $h<0$. In either case, $\partial_t$ is a Killing vector field, and we note that the positive zeroes $r_i$ of $h$ correspond to Killing horizons $\{r=r_i\}$, see below. Both in the study of isolated systems and of cosmology, spherically symmetric spacetimes of class $\mathfrak{H}$, i.e., when $(\mathcal{N},g^{\mathcal{N}})$ is given as the round sphere, yield a large class of models which have been studied extensively, e.g. the Schwarzschild and Reissner-Nordstr\"om spacetimes, and the de\,Sitter and anti-de\,Sitter spacetimes. If we additionally assume that $(\mathcal{N},g^{\mathcal{N}})$ has constant sectional curvature, then a spacetime of class $\mathfrak{H}$ is equipped with a Birmingham--Kottler metric, see e.g. \cite{birmi, kottler, chrugalpot}.
	
	Given a function $h$ as above, we understand the spacetimes $(M_i,g)$ of class $\mathfrak{H}$ with ${M_i=\R\times(r_{i-1},r_i)\times {\mathcal{N}}}$ as different regions of a larger spacetime divided by the Killing horizons $\{r=r_i\}$, an interpretation we will make rigorous with our construction in \Cref{sec_extension}. In line with the usual convention, we denote $M_i$ corresponding to the outermost interval $(r_{i-1},r_i)$ on which $h$ is positive as \emph{Region~\rom{1}} and refer to it as the \emph{domain of outer communication}. Thus, the domain of outer communication corresponds to either $(r_N,r_{N+1})$ or $(r_{N-1},r_{N})$, where in the latter case ``$r=\infty$''  and Region~\rom{1} are separated by a cosmological Killing horizon. As we move inward with respect to the radius, we will denote the spacetimes $(M_i,g)$ corresponding to the open intervals $(r_{i-1},r_i)$ as Regions with an increasing Roman numeral given by a map $L(i)$, where
	\begin{align*}
		L(i)\definedas
		\begin{cases}
			N+2-i&h>0\text{ on }(r_N,r_{N+1}),\\
			N+1-i&h>0\text{ on }(r_{N-1},r_{N}),
		\end{cases}
	\end{align*} 
	for $i\in\{1,\dotsc, N+1\}$. Note that this unconventionally leads to the name Region~$0$ for the region outside a cosmological Killing horizon; this turns out to be convenient due to our iterative definition.
	Suppressing the coordinates on $\mathcal{N}$, we define the \emph{planes} $P_{L(i)}=\R\times(r_{i-1},r_i)$ \emph{with metric coefficient $h$} $P_{L(i)}$ for $1\le i\le N+1$. These are equipped with the induced metric 
	\[
		-h\d t^2+\frac{1}{h}\d r^2.
	\]
 	We will denote the Levi-Civita connection of $(M,g)$ by $\nabla$.
	In a spacetime, the \emph{surface gravity $\kappa$} of a Killing horizon with respect to an ``asymptotically''  timelike Killing vector field $X$ describes the failure of the integral curves of $X$ to be affinely parametrized null geodesics at the horizon. More precisely, $\kappa$ is defined by the equation
		\[
			\nabla_XX=\kappa X
		\]
	evaluated at the horizon. We call said Killing horizon \emph{non-degenerate} if $\kappa\not=0$, \emph{degenerate} if $\kappa=0$. Note that the value of $\kappa$ depends on the scaling of $X$ , so an additional restriction on $X$  is required to define $\kappa$ uniquely. In the asymptotically flat case, one prescribes a natural boundary condition on $X$ at infinity, namely that $g(X,X)\to-1$ as $r\to\infty$ \cite{wald}. Note that in the case of an asymptotically flat spacetime of class $\mathfrak{H}$ (which in particular imposes that $(\mathcal{N},g_{\mathcal{N}})$ is the round sphere), the surface gravity of a Killing horizon $\{r=r_i\}$ w.r.t. $X=\partial_t$ is well-known:
	\begin{lem}\label{lem_surfgrav}\cite[ Equation (12.5.16)]{wald}
		If $h\to 1$ as $r\to\infty$, then the surface gravity $\kappa_N$ w.r.t. $X=\partial_t$ of the Killing horizon $\{r=r_N\}$ bordering Region~\rom{1} satisfies
		\begin{align}\label{eq_surfgrav}
		\kappa_N=\pm\frac{h'(r_N)}{2}.
		\end{align}
	\end{lem} 
	Using $X=\partial_t$, \eqref{eq_surfgrav} holds true for all Killing horizons $\{r=r_i\}$, $i\in\{1,\dotsc,N\}$, in general spacetimes of class $\mathfrak{H}$, as can be seen by the same straightforward computation. As we only need to differentiate between degenerate and non-degenerate Killing horizons, this scaling of $X$ is sufficient for our purposes. In fact, assuming that $(M_i,g)$ admits a generalized Kruskal--Szekeres extension, one finds that $\kappa_i=+\frac{1}{2}h'(r_i)$, see Proposition \ref{prop_surfgrav}.
	
\section{Construction of the generalized Kruskal--Szekeres extension}\label{sec_extension}
	To construct a spacetime extension joining $(M_i,g)$ and $(M_{i+1},g)$ in the spirit of the Kruskal--Szekeres extension, it suffices to show that we can join the planes $P_{L(i)}$, $P_{L(i+1)}$ across their shared boundary $\R\times\{r_i\}$ in a regular way.
	Imitating the approach presented in O'Neill \cite[Pages 386--389]{oneill}, we define the \emph{generalized Kruskal--Szekeres plane} $(\mathbb{P}^i_h,\d s^2)$ as follows.
	\begin{defi}\label{extdefi1}
		Let $h\colon (0,\infty)\to\R$ be a (smooth) function with finitely many zeros ${r_0:=0<r_1<\dotsc<r_N<r_{N+1}:=\infty}$. Assume there exists a (smooth) strictly increasing solution $f_i$ of the ODE 
		\begin{align}\label{eq_ODE_central}
		\frac{f_i}{f_i'}=K_ih
		\end{align}
		on $(r_{i-1},r_{i+1})$ for some $K_i\in\R\setminus\{0\}$, $i\in{1,\dotsc,N}$. We define the generalized Kruskal--Szekeres plane $(\mathbb{P}^i_h,\d s^2)$ (with respect to $f_i$) as
		\begin{align*}
		\mathbb{P}^i_h&\definedas\{(u,v)\in\R^2:uv\in \text{Im}(f_i)\},\\
		\d s^2&=(F_i\circ\rho)(\d u\otimes \d v+\d v\otimes \d u),
		\end{align*}
		where $F_i\definedas\frac{2K_i}{f_i'}$ and $\rho\definedas f_i^{-1}(uv)$.
	\end{defi}
	We will see in Proposition \ref{propisometry} below that the generalized Kruskal--Szekeres plane indeed gives rise to a spacetime extension of $(M_i,g)$ and $(M_{i+1},g)$. Hence, the existence of a generalized Kruskal--Szekeres extension joining $(M_i,g)$ and $(M_{i+1},g)$ solely depends on the solvability of \eqref{eq_ODE_central} for a suitable constant $K_i$.
	For a complete analysis of the ODE \eqref{eq_ODE_central}, we refer to Appendix \ref{app_solving}, but for the convenience of the reader, we state the main result of Appendix \ref{app_solving} directly here:
	\begin{prop}\label{prop_central}
	Let $h$ be as in Definition \ref{extdefi1} Then Equation \eqref{eq_ODE_central} admits a strictly increasing solution $f_i$ on $(r_{i-1},r_{i+1})$ for some $K_i\in\R\setminus\{0\}$ if and only if $h'(r_i)\not=0$. If \eqref{eq_ODE_central} has a solution, $K_i=\frac{1}{h'(r_i)}$ and $f_i$ is uniquely determined up to scaling. Unless otherwise stated, we will choose the unique solution $f_i$ such that $f_i'(r_i)=1$
	\end{prop}
	\begin{bem}\label{bem_central}
		We see from the construction in Appendix \ref{app_solving} that $f_i$ is explicitly given by
		\begin{align}\label{eqexaxtformsolf}
		f_i(r)&=K_ih(r)\exp\left(\frac{1}{K_i}\int\limits_{r_i}^r\frac{1-K_ih'(s)}{h(s)}\d s\right).
		\end{align}
		Assuming that $h\in C^k$ and that $h$ is locally $(k+1)$-times differentiable around $r_i$ for some $k\ge 1$, we see that $\frac{1-K_ih'}{h}$ extends through $r_i$ in $C^k$ by \Cref{lemmaauxilliary1} below, so a-priori $f_i\in C^k$.  However, by the precise formula \eqref{eqexaxtformsolf} for the solution $f_i$, we see that
		\[
		f_i'=\exp\left(\frac{1}{K_i}\int\limits_{r_i}^r\frac{1-K_ih'(s)}{h(s)}\d s\right),
		\]
		so that in fact $f_i\in C^{k+1}$.
	\end{bem}
	By Lemma \ref{lem_surfgrav}, we note that the condition $h'(r_i)\not=0$ is equivalent to the fact that the Killing horizon $\{r=r_i\}$ is non-degenerate. Thus, assuming that the Killing horizon ${\{r=r_i\}}$ is non-degenerate, we know that \eqref{eq_ODE_central} admits a well-defined, strictly increasing (smooth) solution $f_i$ on $(r_{i-1},r_{i+1})$ with $K_i\not=0$ such that $f_i'(r_i)=1$.
		\begin{figure}[H]
		\centering
		\includegraphics[scale=0.7]{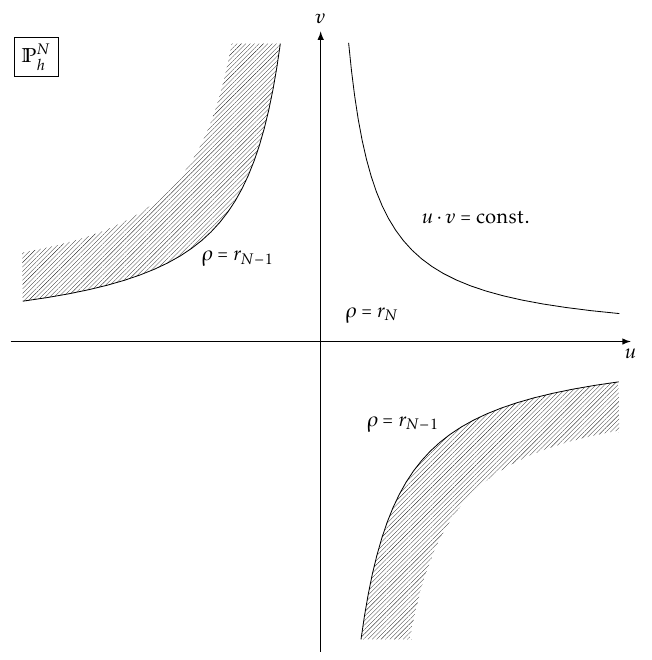}
		\caption{The generalized Kruskal--Szekeres plane $\mathbb{P}^N_h$.}
	\end{figure}
	Next, note that $Im(f_i)$ is an open subset of $\R$ containing $0$, and just as in the Kruskal--Szekeres plane for the Schwarzschild spacetime, the level set curves of $\rho=f^{-1}(uv)$ are hyperbolas $uv=\text{const}.$ if $\rho\not= r_i$, and the coordinate axes if $r=r_i$. Let $Q^i_1,\dotsc,Q^i_4$ be the open quadrants of $\mathbb{P}_h$, where $Q^i_1\definedas\{u,v>0\}$, $Q^i_2\definedas\{v>0, u<0\}$, $Q^i_3\definedas\{u,v<0\}$, and $Q^i_4\definedas\{u>0, v<0\}$.
	We recover the analogous statement to \cite[Lemma 13.23]{oneill}.
	\newpage
	\begin{lem}\label{lem1}
		Let $h$ be as in Definition \ref{extdefi1}, $f_i$ a strictly increasing solution of \eqref{eq_ODE_central} on $(r_{i-1},r_{i+1})$ with $K_i\not=0$, and $(\mathbb{P}^i_h,\d s^2)$ the generalized Kruskal--Szekeres plane with respect to $f_i$. Recalling that $F_i=\frac{2K_i}{f_i'}$ and $\rho(u,v)=f_i^{-1}(uv)$, and defining $\tau\definedas K_i\ln\btr{\frac{v}{u}}$ on the open quadrants $Q^i_1,\dotsc,Q^i_4$ of $\mathbb{P}^i_h$, we find
		\begin{align}
		\begin{split}\label{lemeq1}
				F_if_i&=2K_i^2h,\\
				F_if_i'&=2K_i,\\
				\frac{f_i}{f_i'}&=K_ih,
		\end{split}
		\end{align}
		\begin{align}
		\begin{split}\label{lemeq2}
				\d \tau&=K_i\left(\frac{\d v}{v}-\frac{\d u}{u}\right),\\
				\d \rho&=K_ih\left(\frac{\d u}{u}+\frac{\d v}{v}\right).
		\end{split}
		\end{align}
	\end{lem}
	\begin{proof}
			\eqref{lemeq1} is satisfied by construction. Furthermore, a straightforward computation yields
			\begin{align*}
				\d \tau&=\partial_u\tau\d u+\partial_v\tau\d v
				=K_i\left(\frac{u}{v}\frac{\d v}{u}-\frac{u}{v}\frac{v\d u}{u^2}\right)
				=K_i\left(\frac{\d v}{v}-\frac{\d u}{u}\right),\\
				\d \rho&=\partial_v\rho\d v+\partial_u\rho\d u
				=\frac{u}{f_i'(f_i^{-1}(uv))}\d v +\frac{v}{f_i'(f_i^{-1}(uv))}\d u
				=K_ih\left(\frac{\d v}{v}+\frac{\d u}{u}\right).
			\end{align*}
	\end{proof}
	\begin{prop}\label{propisometry}
		Let $h$ be as in Definition \ref{extdefi1}, and let $P_{L(i+1)}$, $P_{L(i)}$ be the planes with shared boundary $\R\times \{r=r_i\}$. Let $f_i$ be a strictly increasing solution of \eqref{eq_ODE_central} on $(r_{i-1},r_{i+1})$ with $K_i\not=0$, and $(\mathbb{P}^i_h,\d s^2)$ the generalized Kruskal--Szekeres plane with respect to $f_i$. Let $\tau\definedas K_i\ln\btr{\frac{v}{u}}$ where defined, and let $Q^i_1,Q^i_2$ be the first two open quadrants of $\mathbb{P}^i_h$.
		Then the function
		\[
		\psi\colon Q^i_1\cup Q^i_2\to P_{L(i+1)}\cup P_{L(i)}, (u,v)\mapsto(\tau(u,v),\rho(u,v))
		\]
		is a quadrant preserving isometry. Therefore $\mathbb{P}^i_h\times_{\rho^2}\mathcal{N}$ is a spacetime extension  joining $(M_i,g)$ and $(M_{i+1},g)$.
	\end{prop}
	Recall further that the intersection $\{u=v=0\}$ of the (connected and smooth) components $\{u=0\}$ and $\{v=0\}$ of the Killing horizon $\{r=r_i\}$ in the spacetime extension of $M_i$, $M_{i+1}$ is called the \emph{bifurcation surface}.
	\begin{figure}[H]
		\centering
		\includegraphics[scale=0.7]{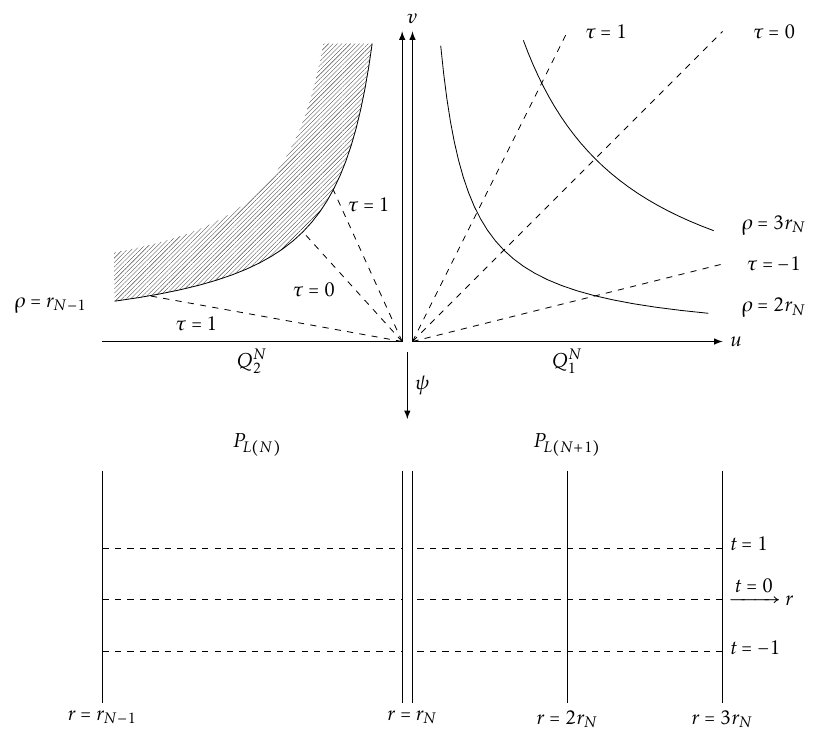}
		\caption{The isometry $\Psi$ mapping $Q_1^N$ into $P_{L(N+1)}$, and $Q_2^N$ into $P_{L(i)}$, respectively.}
	\end{figure}
	\begin{bem}\label{bem_isometry}
	Since $uv=(-u)(-v)$ the map $\Phi\colon\mathbb{P}_h\to\mathbb{P}_h,(u,v)\mapsto(-u,-v)$ is a quadrant interchanging isometry. Therefore, just as it is the case for the Schwarzschild manifold, $\mathbb{P}^i_h$ contains two copies of $P_{L(i+1)}$ and $P_{L(i)}$.
	
	Moreover, note that by the explicit definitions of $\rho$ and $\tau$, one can directly verify that
	\begin{align*}
		u^2&=\btr{f_i(\rho)}\exp\left(-\frac{\tau}{K_i}\right),\\
		v^2&=\btr{f_i(\rho)}\exp\left(+\frac{\tau}{K_i}\right),
	\end{align*}
	which uniquely determines $(u,v)$ on each quadrant $Q^i_1,\dotsc,Q^i_4$.
	\end{bem}
	\begin{proof}
		The fact that $f_i^{-1}$ and $\ln$ are bijective, and that the level sets of $u\cdot v$ and $\btr{\frac{v}{u}}$ intersect in unique points implies that $\psi$ is bijective.
		Furthermore
		\begin{align*}
		\det D\psi_{(u,v)}
		&=-2\frac{K_i}{f'},
		\end{align*}
		and since $h(r)=0$ if and only if $r=r_i$ resp. $uv=0$, we have $\det D\psi_{(u,v)}\not=0$ on $Q_1\cup Q_2$. 
		
		Therefore $\Psi$ is a diffeomorphism and in fact an isometry, since
		\begin{align*}
		\Psi^*\left(-h\d t^2+\frac{1}{h}\d r^2\right)
		=&-\left(h\circ\rho \right)\d \tau^2+\frac{1}{\left(h\circ\rho \right)}\d \rho^2 \\
		=&-K_i^2\cdot\left(h\circ\rho \right)\left(\frac{\d u}{u}-\frac{\d v}{v}\right)^2+\frac{K_i^2\left(h\circ\rho \right)^2}{\left(h\circ\rho \right)}\left(\frac{\d u}{u}+\frac{\d v}{v}\right)^2 \\ 
		=&\frac{2K_i^2\cdot\left(h\circ\rho \right)}{\left(f_i\circ\rho \right)}(\d u\otimes \d v+\d v\otimes \d u) \\
		=&(F_i\circ\rho)(\d u\otimes \d v+\d v\otimes\d u).
		\end{align*}
		
		Lastly $f_i'>0$, so $\rho(u_1,v_1)>\rho(u_2,v_2)$ if and only if $u_1v_1>u_2v_2$. Since $f_i(r_i)=0$, it holds that $\rho(u,v)>r_i$ for $uv>0$ and $\rho(u,v)<r_i$ for $uv<0$. Hence, $\Psi(Q^i_2)=P_{L(i)}$ and $\Psi(Q^i_1)=P_{L(i+1)}$. This concludes the proof.
	\end{proof}
	We will henceforth call the resulting spacetime extension a \emph{generalized Kruskal--Szekeres extension}, and have seen that a spacetime of class $\mathfrak{H}$ admits such an extension across a Killing horizon if and only if the Killing horizon is non-degenerate. We can further directly compute the surface gravity $\kappa_i$, $1\le i\le N$, in the \emph{double null coordinates} $u$, $v$.
	\begin{prop}\label{prop_surfgrav}
		Let $(M_i,g)$, $1\le i\le N$ be a spacetime of class $\mathfrak{H}$ admitting a generalized Kruskal--Szekeres extension across the Killing horizon $\{r=r_i\}$. Then the surface gravity $\kappa_i$ of $\{r=r_i\}$ is given by 
		\[
		\kappa_i=\frac{1}{2K_i}=\frac{h'(r_i)}{2}.
		\]
	\end{prop}
	\begin{proof}
		We compute $\nabla_{\partial_t}\partial_t$ in the global null coordinates $u,v$ introduced in Definition \ref{extdefi1} and use the properties of the generalized Kruskal--Szekeres extension stated in Lemma \ref{lem1}. In this coordinates
		\[
		\partial_t=-h\cdot\grad \tau=\frac{1}{2K_i}\left(v\partial_v-u\partial_u\right).
		\]
		The Killing horizon corresponds to the null hypersurface $\{u=0,v>0\}$, therefore 
		\[
		\partial_t=\frac{v}{2K_i}\partial_v
		\]
		at the Killing horizon.
		A straightforward computation yields
		\[
		\nabla_{\partial_t}\partial_t=\frac{1}{4K_i^2}\left(1-K_ih\frac{f_i''}{f_i'}\right)(v\partial_v+u\partial_u),
		\]
		so at the horizon, where $u,h=0$, we get
		\[
		\nabla_{\partial_t}\partial_t=\frac{1}{2K_i}\partial_t.
		\]
		This concludes the proof, as we know that $K_i=\frac{1}{h'(r_i)}$ by Proposition \ref{prop_central}.
	\end{proof}
	We summarize the above into our first main result:
	\begin{thm}\label{thm_main1}
		Let $h\colon (0,\infty)\to\R$ be a (smooth) function with finitely many zeros\linebreak ${r_0:=0<r_1<\dotsc<r_N<r_{N+1}:=\infty}$, and let $(\mathcal{N},g_{\mathcal{N}})$ be an $(n-1)$-dimensional Riemannian manifold, $n\ge 3$. Then, for $1\le i\le N$, the spacetimes $(M_i,g)$, $(M_{i+1},g)$ of class $\mathfrak{H}$ can be joined across the Killing horizon $\{r=r_i\}$ by a generalized Kruskal--Szekeres extension if and only if the Killing horizon $\{r=r_i\}$ has non-vanishing surface gravity $\kappa_i=\frac{h'(r_i)}{2}\not=0$.
		The extension is fully determined by the unique, strictly increasing solution $f_i$ of \eqref{eq_ODE_central} with $K_i=\frac{1}{2\kappa_i}$ such that $f'(r_i)=1$, and $f_i$ and the metric coefficient $F_i$ satisfy 
		\begin{align*}
		\frac{f_i}{f_i'}&=K_ih,\\
		F_if_i'&=2K_i,\\
		F_if_i&=2K_i^2h,
		\end{align*}
		where $f_i=uv$ at a point $(u,v,p)\in \mathbb{P}^i_h\times\mathcal{N}$.\newline
		If $h\in C^k(r_{i-1},r_{i+1})$ and $(k+1)$-times differentiable locally around $r_i$ for some $k\ge1$, then ${f\in C^{k+1}(r_{i_1},r_{i+1})}$, so that the metric $g$ extends in $C^k$ across the Killing horizon $\{r=r_i\}$.
	\end{thm}
	Assuming now that any positive zero $r_i$ of the function $h:(0,\infty)\to\R$ is simple, i.e., $h'(r_i)\not=0$ for all $i=1,\dotsc,N$, we can join any two spacetimes of class $\mathfrak{H}$ corresponding to the planes $P_{L(i+1)}$, $P_{L(i)}$ with metric coefficient $h$ and shared boundary $\R\times\{r_i\}$ by constructing the generalized Kruskal--Szekeres extension corresponding to the generalized Kruskal--Szekeres plane $\mathbb{P}_h^i$, containing the quadrants $Q_1^i,\dotsc,Q_4^i$, with respect to the unique strictly increasing solution $f_i$ satisfying
	\[
	\frac{f_i}{f_i'}=K_ih
	\]
	with $f'(r_i)=1$ and $K_i\definedas\frac{1}{h'(r_i)}$. Going forward, we will omit the index $i$ for the sake of simplicity wherever confusion seems unlikely. Moreover, we will from now on join all $(M_i,g)$ into a (disconnected) spacetime $(M,g)$. Hence, there exists a spacetime extension containing all positive radii which is covered by a countable atlas which is regular across any non-degenerate Killing horizon.
	\begin{thm}\label{thm_main2}
		Let $k\ge 1$. Let $(M,g)$ be a spacetime of class $\mathfrak{H}$ with metric coefficient $h$ and fibre $(\mathcal{N},g_{\mathcal{N}})$, such that $h\in C^k(0,\infty)$ and $h$ is $(k+1)$-times differentiable locally around its positive, simple zeros $r_1,\dotsc,r_N$. Then $(M,g)$ admits a (connected) spacetime extension $(\widetilde{M}, \widetilde{g})$, such that $\widetilde{M}$ is covered by a countable $C^{k+1}$-atlas, where each chart is fully determined by a strictly increasing solution $f_i$ of \eqref{eq_ODE_central} on $(r_{i-1},r_{i+1})$, and the metric $\widetilde{g}$ in each chart is $C^k$ across their respective non-degenerate Killing horizon $\{r=r_i\}$.
	\end{thm}
	\begin{bem}\label{bem_main2}
		By the nature of our approach, it is easy to see that the construction readily extends to general warped product metrics of the form
		\[
			-h(r)\d t^2+\frac{1}{h(r)}\d r^2+\omega(r)g_\mathcal{N},
		\]
		where $\omega$ is a positive function on $(0,\infty)$. For example, Gibbons--Maeda--Garfinkle--Horowitz--Strominger (GMGHS) dilation black hole model is of the above form with $\omega'\not=0$, see \cite{gibbonsmaeda, garfinklehorowitzstrominger}. For spacetimes of the above form with $\omega'\not=0$, one can perform a change of the radial coordinate such that the metric is of the form
		\[
			-q(s)\d t^2+\frac{1}{p(s)}\d s^2+s^2g_\mathcal{N},
		\]
		where $s$ coincides with the volume radius and $p=aq$ for some strictly positive function $a=a(s)$, which implies that the zeros of $p$ and $q$ coincide. Metrics of this form seem to play a role in the study of effective one-body mechanics, c.f. \cite{BD}, although in general not under the assumption that the zeros of $p$ and $q$ coincide. It is unclear to us whether one can construct a spacetime extension in a similar manner if the above condition $p=aq$ is violated.
	\end{bem}
	\begin{bem}\label{bem_main2_2}
		One might also think of considering even more general metrics of the form
		\[
		-h(r)\d t^2+\frac{1}{h(r)}\d r^2+\omega(t,r)g_\mathcal{N},
		\]
		where $\omega$ is a positive function on $\R\times(0,\infty)$, which satisfies certain conditions as $r\to r_i$ that ensure that $\omega$ glues smoothly across the Killing horizon in $(u,v)$-coordinates (or at least as regular as $h$). Of course, in this setting the zeros $r_i$ no longer correspond to Killing horizons in general. It is easy to see that the above construction still works provided that $\omega$ is independent of $t$ near the horizon, and has extremely high falloff rates as $t\to\pm\infty$.
	\end{bem}
	Recall from Remark \ref{bem_isometry} that any generalized Kruskal--Szekeres plane $\mathbb{P}_h^i$ corresponding to the respective solution $f_i$ on $(r_{i-1},r_{i+1})$ contains two copies of the planes $P_{L(i+1)}$, $P_{L(i)}$. Thus, unless $N=1$, $(\widetilde{M}, \widetilde{g})$ in fact contains infinitely many, countable copies of each Region.
	
	We can endow $(\widetilde{M},\widetilde{g})$ with a time-orientation in the following way: If $N=1$, $(\widetilde{M},\widetilde{g})$ is covered by a single chart and either $\partial_v-\partial_u$ or $\partial_v+\partial_u$ is a global timelike vector, depending on the sign of $K$, and the time-orientation of $\partial_t$ within a chosen copy of domain of outer communication extends to all of $(\widetilde{M},\widetilde{g})$. We will adopt the notation described below also for this case.
	
	Now assume that $N>1$. Let $1\le i\le N$ be such that $L(i+1)=1$. Pick any copy of the generalized Kruskal--Szekeres plane $\mathbb{P}_h^i$. By the choice of $i$, one sees that $K_i>0$ and that $\mathbb{P}_h^i$ contains two copies of the domain of outer communication Region~\rom{1}, where $\partial_t=\frac{1}{2K_i}(v\partial_v-u\partial_u)$ is timelike. We define $\partial_t$ to be future-pointing in the copy of Region~\rom{1} corresponding to the quadrant $Q^i_1=\{u,v>0\}$ and denote it henceforth as Region~\rom{1}+. As $K_i>0$, we note that 
	\[
		\widetilde{g}(\partial_v-\partial_u,\partial_v-\partial_u)=-\frac{4K_i}{f_i'}<0,
	\]
	and
	\[
		\widetilde{g}(\partial_t,\partial_v-\partial_u)=-\frac{1}{f_i'}(v+u),
	\]
	so $\partial_v-\partial_u$ is timelike, future-pointing everywhere on $\mathbb{P}_h^i$. Observe that $\partial_t$ is past-pointing on the copy of Region~\rom{1} corresponding to $\{u,v<0\}$, and we will denote it henceforth as Region~\rom{1}$-$. Since $K_i>0$, $\mathbb{P}_h^i$ further contains two copies of Region~\rom{2} corresponding to the quadrants $\{u>0,v<0\}$, $\{u<0,v>0\}$, where $\partial_r=\frac{1}{2K_ih}(v\partial_v+u\partial_u)$ is timelike. We observe that $\partial_r$ is past-pointing on $\{v>0,u<0\}$ and future-pointing on $\{v<0,u>0\}$, which we will denote as Region~\rom{2}$+$ and Region~\rom{2}$-$, respectively. 

	We can then extend this choice of time-orientation iteratively to all of $(\widetilde{M},\widetilde{g})$ in the following way: First note by way of construction that any chart corresponding to generalized Kruskal--Szekeres coordinates $(u,v)$ overlaps with at least one other such chart in a region where either $\partial_t$ or $\partial_r$ is timelike everywhere in this region. We may then assume that the time-orientation of $\partial_t$ or $\partial_r$ is already determined via the overlapping coordinate system. Without loss of generality, we may associate the coordinate chart with a copy of the generalized Kruskal--Szekeres plane $\mathbb{P}^i_h$ containing two copies of both $P_{L(i)}$ and $P_{L(i+1}$, for some $1\le i\le N$, associated to Regions with the corresponding Roman numerals, and we differentiate between two cases: If $L(i+1)$ is odd, we have $K_i>0$ and extend the time-orientation to $\mathbb{P}^i_h$ via the timelike vector field $\partial_v-\partial_u$, while if $L(i+1)$ is even, we have $K_i<0$ and extend the time-orientation to $\mathbb{P}^i_h$ via the timelike vector field $\partial_v+\partial_u$. We add a $+$ to a Roman numeral, if either $\partial_t$ is timelike and future-pointing, or else if $\partial_r$ is timelike and past-pointing. Otherwise, we add a $-$ to the Roman numeral. Note that we attach a ``new''  copy of the respective Kruskal--Szekeres plane at each step in our iterative process which has not yet been endowed with a time-orientation, as we want to consider $(\widetilde{M},\widetilde{g})$ to be ``maximal''. In this way, the time-orientation on $(\widetilde{M},\widetilde{g})$ will be well-defined by virtue of the construction.
	
	Lastly, we obtain a globally defined, future timelike vector field by a linear combination of all the locally defined vector fields with respect to an appropriate choice of partition of unity of $(0,\infty)$. Compare the figure below for a Penrose--Carter diagram of the generalized Kruskal--Szekeres spacetime if $N=2$. This distinction of the difference in time-orientation of different copies of the same region is in fact important for the discussion of symmetric photon surfaces in \Cref{sec_photonsurf}. 
\newpage

	\begin{figure}[H]
		\centering
		\includegraphics[scale=0.75]{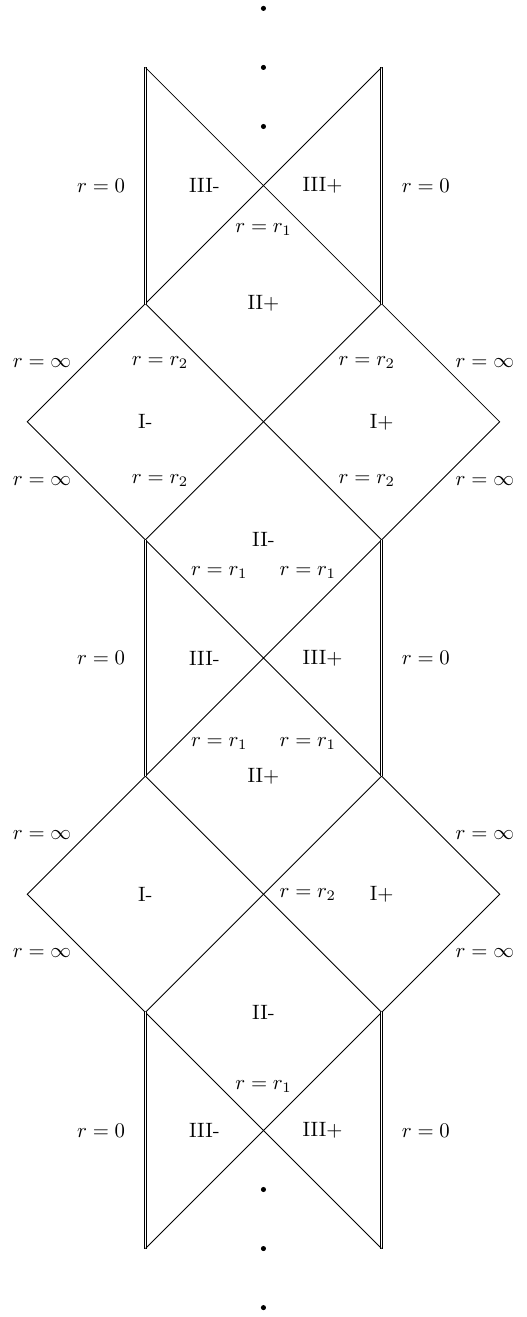}
		\caption{The generalized Kruskal--Szerekes extension of the subextremal Reissner-Nordstr\"om spacetime.}
	\end{figure}
\newpage
	Let us now touch on the topic of maximality and geodesic (in-) completeness.
	\begin{prop}\label{prop_kretschmann}
		Let $(M,g)$ be a spacetime of class $\mathfrak{H}$ with metric coefficient $h$ and fibre $(\mathcal{N},g_{\mathcal{N}})$, and let $(\widetilde{M},\widetilde{g})$ be the corresponding generalized Kruskal--Szekeres spacetime. The Kretschmann scalar $\widetilde{\mathcal{K}}\definedas\vert{\widetilde{\operatorname{Rm}}}\vert^2$ of $(\widetilde{M},\widetilde{g})$ is given as
		\begin{align}
			\begin{split}\label{eq_kretschmann}
			\widetilde{\mathcal{K}}=&\,(h'')^2+2(n-1)\frac{(h')^2}{r^2}+\frac{2(n-1)(n-2)}{r^4}\left(\frac{\operatorname{R}_\mathcal{N}}{(n-1)(n-2)}-h\right)^2\\
			&\,+\frac{1}{r^4}\left(\mathcal{K}_\mathcal{N}-\frac{2\operatorname{R}_\mathcal{N}^2}{(n-1)(n-2)}\right),
			\end{split}
		\end{align}
		where $\mathcal{K}_\mathcal{N}$ and $\operatorname{R}_\mathcal{N}$ denote the Kretschmann scalar and scalar curvature of the fibre $(\mathcal{N},g_\mathcal{N})$, respectively.
	\end{prop}
	\begin{proof}
		By continuity, we can perform all computations in the open regions away from the horizon in $(t,r)$-coordinates. The claim then follows by a standard computation.
	\end{proof}
	Since each of the four terms  in \eqref{eq_kretschmann} is manifestly non-negative, it is immediate to see that $\widetilde{\mathcal{K}}\to\infty$ as $r\to\infty$ unless possible when $n=3$ or when $(\mathcal{N},g_\mathcal{N})$ is a metric of constant sectional curvature. Thus, the Birmingham--Kottler metrics where $(\mathcal{N},g_\mathcal{N})$ is a metric of constant sectional curvature cf. \cite{birmi, kottler}, seem the most relevant examples to discuss (in-)extendability of the generalized Kruskal--Szekeres spacetimes, in particular in higher dimensions.
	
	A complete discussion about (in-)extendability and geodesic (in-)completeness of a generalized Kruskal--Szekeres spacetime would be beyond the scope of this paper. However, Proposition \ref{prop_kretschmann} gives some direct criteria for $C^2$-inextendability:
	\begin{kor}\label{kor_inextendability}
		Let $(M,g)$ be a spacetime of class $\mathfrak{H}$ with metric coefficient $h$ and fibre $(\mathcal{N},g_{\mathcal{N}})$, and let $(\widetilde{M},\widetilde{g})$ be the corresponding generalized Kruskal--Szekeres spacetime. Assume that $h''$ or $rh'$ or
		\[
			r^2\left(h-\frac{\operatorname{R}^\mathcal{N}}{(n-1)(n-2)}\right)
		\]
		are unbounded as $r\to0$.
		Then $(\widetilde{M},\widetilde{g})$ is $C^2$-inextendable. \newline
		If further 
		\[
			\int\limits_0^{r_1}\frac{1}{\sqrt{\btr{h(r)}}}\d r<\infty,
		\]
		then $(\widetilde{M},\widetilde{g})$ is geodesically incomplete.
	\end{kor}
	\begin{proof}
		As a coordinate independent scalar given by second derivatives of the metric, the Kretschmann scalar would be $C^0$ across $r=0$ for any $C^2$-extension of $(\widetilde{M},\widetilde{g})$. In particular, it would remain bounded as $r\to 0$. As $\widetilde{\mathcal{K}}\to\infty$ for $r\to 0$ under the assumptions of this corollary by Proposition \ref{prop_kretschmann}, no $C^2$-extension can exist.
		
		Now assume that
		\[
		\int\limits_0^{r_1}\frac{1}{\sqrt{\btr{h(r)}}}\d r<\infty.
		\]
		Fix any $p\in\mathcal{N}$ and consider the regularly parametrized curve $\gamma\colon(-u_0,u_0)\to \mathbb{P}_h^1\times \mathcal{N}\subseteq \widetilde{M}$ with $\gamma(s)=(u(s),v(s),p):=(s,-s,p)$, where $u_0>0$ is determined by $f_1^{-1}(-u^2)\to0$ as $u\to u_0$. Note that by this choice of $u_0$ $\gamma$ is inextendable in $\widetilde{M}$. A direct computation shows that $\widetilde{g}(\dot\gamma,\dot\gamma)=-2F_1\circ f_1^{-1}(-s^2)\not=0$ and
		\[
			\widetilde{\nabla}_{\dot\gamma}\dot\gamma=\dot\gamma
		\]
		for some function $b$ along $\gamma$. In particular, we can reparametrize $\gamma$ as the geodesic $\widetilde{\gamma}$ in $\widetilde{M}$ satisfying $\gamma(0)=(0,0,p)$, $\dot\gamma(0)=\partial_u-\partial_v$. To prove geodesic incompleteness, it remains to show that the length of $\gamma$ is finite. As both $\gamma\vert_{(0,u_0)}$ and $\gamma\vert_{(u_0,0)}$ can be identified with a radial curve in the set $\{t=0\}$ in $\mathbb{P}_h^1$, a direct computation gives that
		\[
			L[\gamma]=2\int\limits_0^{r_1}\frac{1}{\sqrt{\btr{h(r)}}}\d r<\infty
		\]
		by assumption.
	\end{proof}

\section{Construction of a global tortoise function}\label{sec_globaltortoisefunc}
	In their recent paper \cite{Schindler2018AlgorithmsFT}, Schindler and Aguirre numerically implemented an algorithm for the construction of Penrose--Carter diagrams for spherically symmetric spacetimes of a \emph{class SSS}, which corresponds to class $\mathfrak{H}$ in spherical symmetry, i.e., for $(\mathcal{N},g_\mathcal{N})$ being the round sphere. More specifically, they use an algorithm to construct global Penrose coordinates in which the metric extends continuously across the non-degenerate Killing horizon. To do so, they construct a \emph{global tortoise function} $R^*$, i.e., a primitive of $\frac{1}{h}$. Their tortoise function $R^*$ is well-defined on all of $(0,\infty)$ except at the finitely many roots $r_i$ of $h$, where $R^*$ satisfies Equation \eqref{eqbrillhaywardanalogue} (as stated in Appendix \ref{app_solving}) in a neighborhood of each $r_i$ simultaneously up to a possible additive constant $c_i$ at each root. They show that this yields a metric in ``Kruskal-like coordinates'' in the same manner as the construction by Brill and Hayward in \cite{brillhayward}.\newline
	More precisely, their algorithm relies on the fact that the tortoise function $R^*$, which is well-defined on $(0,\infty)\setminus\{r_i\}_{i=1}^N$, satisfies
	\[
	\lim_{r\to r_i}\left(R^*(r)-k_i\ln(\btr{r-r_i})\right)=c_i,
	\]
	for constants $c_i\in\R$, where $k_i\definedas h'(r_i)$. The metric then takes the form
	
	\[
	g=\frac{4}{k_i^2}\btr{h(r)}\exp\left(-k_iR^*(r)\right)\d U\d V+r^2\d\Omega,
	\]
	with $\btr{UV}=\exp(k_iR^*(r))$. In their approach, the global tortoise function $R^*$ is obtained as the limit of complex path integrals over $\frac{1}{h}$ along a contour line which avoid the roots $r_i$ by semicircles of arbitrarily small radius. For the path integrals to be well-defined along the small semicircles and to conclude the above properties of the tortoise function, Schindler--Aguirre impose real analyticity of $h$ at each horizon radius $r_i$. Away from the horizon radii $r_i$, they impose rather mild regularity conditions, assuming $h$ to be only differentiable. Our above analysis requires $h\in C^1$ and $h$ twice differentiable near the horizon radii $r_i$. This is clearly a stronger global assumption, but of course a significantly weaker one near the Killing horizons. Hence, with the additional assumption $h\in C^1$, we can generalize the construction by Schindler--Aguirre (to class $\mathfrak{H}$) by constructing a global tortoise function with the desired properties in our setting. This can be seen as follows:
	
	By Theorem \ref{thm_main2}, we get a strictly increasing solution $f_i$ of \eqref{eq_ODE_central} with constant $K_i=\frac{1}{h'(r_i)}$ on $(r_{i-1},{r+1})$ for each $i\in\{1,\dotsc,N\}$ which is at least $C^2$, depending on our assumptions on $h$. Then the function $R^*_i\definedas K_i\ln(\btr{f_i})$ is well-defined and a primitive of $\frac{1}{h}$ on $(r_{i-1},r_i)\cup(r_i,r_{i+1})$. However, by the fundamental theorem of calculus, for each $i\in\{1,\dotsc, N\}$, there exists a constant $C_i$ depending only on the solutions $f_1,\dotsc, f_N$, on $h$, and on a possible global constant of integration, such that the function
	\begin{align}
	\begin{split}\label{eq_globaltortoisefunc}
	&R^*\colon(0,\infty)\setminus\{r_1,\dotsc, r_N\}\to \R,\\
	&r\mapsto K_i\ln(\btr{f_i(r)})+C_i=K_i\ln(\btr{K_ih})+\int\limits^r_{r_i}\frac{1-K_ih'}{h}\d s+C_i,\text{ if }r\in (r_{i-1},r_{i+1})\setminus\{r_i\}
	\end{split}
	\end{align}
	is well-defined and at least $C^2$. We thus a posteriori recover a global tortoise function with the same properties as that of Schindler and Aguirre, since
	\[
	R^*(r)-K_i\ln(\btr{r-r_i})\to C_i\text{ as }r\to r_i.
	\] 
	Our construction is converse (up to a factor) to the approach of Brill and Hayward, since 
	\[
	\btr{UV}=\exp(h'(r_i)R^*)=\exp(C_ih'(r_i))\btr{f_i},
	\]
	and the metric coefficient satisfies
	\[
	g_{UV}=\frac{2}{h'(r_i)^2}\btr{h(r)}\exp(-h'(r_i)R^*(r))=\exp(-h'(r_i)C_i)F_i.
	\]
	The constant of integration $C_i$ corresponds to a rescaling of $f_i$ by a factor $\exp(C_ih'(r_i))$ within the $1$-parameter class of solutions. It is not surprising that this rescaling is in general necessary, as each tortoise function $R^*_i$ initially corresponds to the solution $f_i$ with ${f'_i(r_i)=1}$. Note however, that this matching up to a constant only works on the level of tortoise functions and not on the level of solutions $f_i$ of \eqref{eq_ODE_central}: Even rescaled and sign-switched solutions $f_{i-1}$ and $f_i$ will not coincide where they overlap as they necessarily solve different ODEs.
	
	By its definition, global tortoise function must be unique up to an additive constant, hence we recover the global tortoise function of Schindler and Aguirre. In view of numerical implementation, the effort of computing the global tortoise function via \eqref{eq_globaltortoisefunc} seems at most comparable to computing it via the complex contour integrals by Schindler and Aguirre. Our results assert that this algorithm converges as long as $h\in C^1$ and $h$ is twice differentiable near each $r_i$, for arbitrary fibre $(\mathcal{N},g_{\mathcal{N}})$.

\section{Photon surfaces in the generalized Kruskal--Szekeres spacetime}\label{sec_photonsurf}
	We consider photon surfaces in a generalized Kruskal--Szekeres spacetime $(\widetilde{M},\widetilde{g})$ with respect to the metric coefficient $h\colon (0,\infty)\to\R$ with finitely many, simple zeros $r_i$ ($1\le i\le N$) and fibre $(\mathcal{N},g_{\mathcal{N}})$. In this section, we will perform all local computations in a $(u,v)$ coordinate patch of $(\widetilde{M},\widetilde{g})$, on which we have a (unique) strictly increasing solution $f$ of \eqref{eq_ODE_central} that determines the radial coordinate $\rho(u,v)=f^{-1}(u\cdot v)$ and metric coefficient $F(r)=\frac{2K}{f'(r)}$. Recall that from the coordinate axis, the time coordinate is given by $\tau(u,v)=K\ln\btr{\frac{v}{u}}$. 
	
	A \emph{photon surface} $P^n$ is defined as a null totally geodesic timelike (connected) hypersurface. These are of interest in geometric optics and more generally for understanding trapping phenomena, see e.g. \cite{cedgal, CVE, Perlick}. Under the assumption that $(\mathcal{N},g_\mathcal{N})$ is a round sphere, ``spherically symmetric'' photon surfaces $P^n$ in a spacetime of class $\mathfrak{H}$ are characterized by their ``radial profile''  satisfying a certain ODE in the domain of outer communication of a non-degenerate black hole in $r$-$t$-coordinates (\cite[Theorem 3.5]{cedgal}). Existence and solutions to said ODE is extensively discussed in \cite{cedoliviasophia} by Cederbaum--Jahns--Vi\v{c}\'{a}nek-Mart\'{i}nez. Considering also translations in time (and time-reflections), solutions in the domain of outer communication of a non-degenerate black hole look as follows:
	
	\begin{figure}[h!]
		\centering
		\includegraphics[scale=0.65]{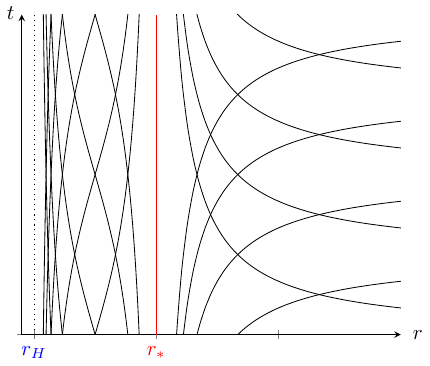}
		\includegraphics[scale=0.65]{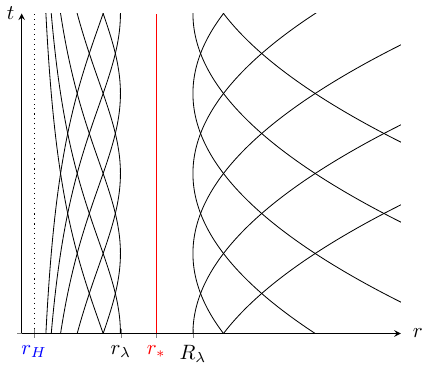}
		\includegraphics[scale=0.65]{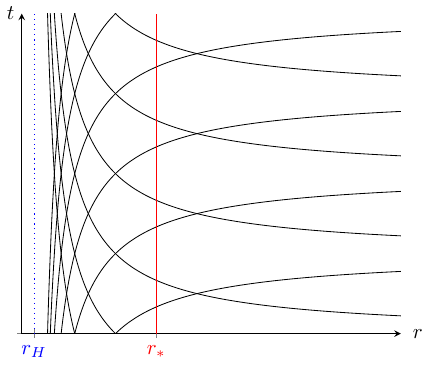}
		\caption{Examples of photon surfaces. The \textcolor{red}{red vertical line} represents a photon sphere. The \textcolor{blue}{dotted line} represents the Killing horizon. \cite[Figure 2]{cedoliviasophia}.}
		\label{fig:casesintro}
	\end{figure}
	
	The aim of this section is to extend the analysis of \cite{cedgal} and \cite{cedoliviasophia} across the horizon and to get a refined understanding of their asymptotic behavior in the asymptotically flat case ($h\to 1$ as $r\to\infty$). We also consider the non-spherical case $\mathcal{N}\not=\mathbb{S}^{n-1}$. In this context, we say a photon surface $P^n$ in $\widetilde{M}$ is \emph{symmetric}, if $P^n$ is generated by a future timelike \emph{profile curve} $\gamma\colon I\to \widetilde{M}$ with $\dot{\gamma}=\dot\gamma^u\partial_u+\dot\gamma^v\partial_v$, i.e.
	\[
		P^n=\{(u,v,p)\colon u=\gamma^u(s),\text{ }v=\gamma^v(s),\text{ }p\in\mathcal{N}\}
	\]
	in all $(u,v)$ coordinate patches, generalizing \cite[Definition 3.3]{cedgal}. In particular, any choice of a spacelike unit normal $\eta$ to $P^n$ is of the form $\eta=a\partial_u+b\partial_v$ (in local $(u,v)$ coordinates), and we note that the tangent space of $\mathcal{N}$ at a point $p\in \mathcal{N}$ is contained in the tangent space of $P^n$ for every point $(u,v,p)\in P^n$. Furthermore, we will identify $\gamma$ with a curve in $\mathbb{P}_h$ whenever convenient. We say $P^n$ is a \emph{symmetric photon sphere}, if $\gamma^u\cdot \gamma^v$ is constant on $I$, i.e., $\rho$ is constant along $P^n$.
	
	We immediately observe the following fact, naturally extending the spherically symmetric case. 
	\begin{prop}\label{prop_constantumbillicalfactor}
		Let $P^n$ be a symmetric photon surface in $(\widetilde{M},\widetilde{g})$ with umbilicity factor $\lambda$. Then $\lambda$ is constant along $P^n$.
	\end{prop}
	\begin{proof}
		Let $e_1,\dotsc, e_{n-1}$ denote a local ON frame of $T\mathcal{N}$ along $P^n$. We assume without loss of generality that $\gamma$ is parametrized by proper time and denote $e_0\definedas \dot{\gamma}$. Since $e_0=\alpha\partial_u+\beta\partial_v$ for $\alpha:=\dot\gamma^u$, $\beta:=\dot\gamma^v$, $\eta$ must necessarily satisfy
		\[
		\eta=\pm\left(\beta\partial_v-\alpha\partial_u\right)
		\]
		with uniform sign along on $P^n$. Then using the explicit form of the Ricci curvature tensor as in Proposition \ref{prop_curvatureextension} and the umbilicity of $P^n$ with umbilicity factor $\lambda$, the Codazzi equations imply
		\begin{align*}
			(n-1)\nabla_{e_I}\lambda&=\Ric(\eta,e_I)=0,\\
			(n-1)\nabla_{e_0}\lambda&=\Ric(\eta,e_0)=(\alpha\beta-\beta\alpha)\Ric_{uv}=0.
		\end{align*}
		Therefore $\lambda$ is constant along $P^n$.
	\end{proof}
	In analogy with Cederbaum--Galloway \cite{cedgal}, we will show that symmetric photon surfaces are fully characterized by a system ODEs for the profile curve $\gamma$, and we will show that this system of ODEs coincides with the that in \cite{cedgal} on the original manifold corresponding to Region~\rom{1}+.
	\begin{prop}\label{prop_charakterizationsymmetricphotonsurfaces}
		Let $P^n$ be a symmetric photon surface with future directed profile curve $\gamma\colon I\to \widetilde{M}$, $\gamma(s)=(u(s),v(s),p)$, parametrized by proper time and with umbilicity factor $\lambda$ with respect to the choice of unit normal $\eta=\dot{v}\partial_{v}-\dot{u}\partial_{u}$. Then the following system of ODEs is satisfied along $\gamma$:
		\begin{align}
			\begin{split}\label{symmetricODEs}
			\lambda&=\frac{1}{\rho f'(\rho)}\left(\dot{v}u-\dot{u}v\right)\\
			\left(\dot{\rho}\right)^2&=\rho^2\lambda^2-h(\rho).
			\end{split}
		\end{align}
		Conversely, let $\gamma$ be a future directed timelike curve $\gamma$ with $\dot{\rho}\not=0$ everywhere along $\gamma$ such that $\dot\gamma$ is orthogonal to $\mathcal{N}$ along $\gamma$. If $\gamma$ satisfies the first order system \eqref{symmetricODEs} for some constant $\lambda$ then $\gamma$ is the profile curve of a symmetric photon surface with umbilicity factor $\lambda$.\newline
		Additionally, a symmetric photon surface $P^n$ is a \emph{photon sphere}, with $\rho=\rho_*$ along $P^n$, if and only if \eqref{symmetricODEs} is satisfied for a profile curve $\gamma$ as above with $\rho\vert_{\gamma}=\rho_*$ and $\rho_*$ is a critical point of $\frac{h}{r^2}$.
	\end{prop}
		\begin{bem}
		As in the spherically symmetric case (\cite[Remark 3.15]{cedgal}), $\lambda$ will in general not be constant along symmetric photon surfaces in generalized Kruskal--Szekeres extensions of spacetimes with more general metrics as considered in Remark \ref{bem_main2}. Nonetheless, Cederbaum--Senthil Velu \cite{cederbaumsenthilvelu} show that there exists a scalar function depending on the umbilicity factor $\lambda$,  as well as on the metric coefficient $\omega$ and its derivative, such that this function is constant along any symmetric photon surface. This scalar reduces to $\lambda$ if $\omega=r^2$. Cederbaum and Senthil Velu show that by replacing $\lambda$ with this more general constant function, one recovers the ODE system as in \cite{cedgal}. Furthermore the analysis in \cite{cedoliviasophia} directly extends to such spacetimes via the aforementioned generalization. In particular, the extension of the analysis in \cite{cedoliviasophia} across the Killing horizons presented in this section should easily generalize to the setting of Remark \ref{bem_main2}. 
	\end{bem}
	\begin{proof}
		Let $\gamma(s)=(u(s),v(s),p)$, $s\in I$, $p\in\mathcal{N}$, denote the future directed timelike profile curve of a symmetric photon surface $P^{n}$ in a $(u,v)$-coordinate patch of a generalized Kruskal--Szekeres spacetime with ${\dot{\gamma}=\dot{u}\partial_{u}+\dot{v}\partial_{v}}$. Assume that $\gamma$ is parametrized by proper time, i.e.,
		\begin{align}\label{eq_arclength}
		2(F\circ \rho)(uv)\,\dot{u}\dot{v}&=-1.
		\end{align}
		This implies that $\dot u\not=0\not=\dot v$ everywhere along $\gamma$. 
		We extend $e_{0}\definedas\dot{\gamma}$ to a local orthonormal tangent frame  $\lbrace{e_{0},e_{J}\rbrace}_{J=1}^{n-1}$ for $P^{n}$, where $\lbrace{e_{J}\rbrace}_{J=1}^{n-1}$ is a (local) ON-frame of $T\mathcal{N}$ along $P^n$ as before. By assumption, we consider the unit normal $\eta$ to $P^{n}$ given by
		\begin{align}\label{eq_eta}
		\eta=&\dot{v}\partial_{v}\vert_{\gamma}-\dot{u}\partial_{u}\vert_{\gamma}
		\end{align}
		along $\gamma$, which one can directly verify to be indeed orthogonal to $P^n$. A direct computation using Proposition \ref{prop_curvatureextension} shows that 
		\begin{align}
		\nabla_{e_{J}}\eta&=\frac{1}{\rho f'(\rho)}\left(u\dot{v}-\dot{u}v\right)e_{J}
		\end{align}
		for all $J=1,\dots,n-1$ so that the second fundamental form $\mathfrak{h}$ of $P^{n}$ in $(\widetilde{M},\widetilde{g})$ reduces to
		\begin{align}
		\mathfrak{h}(e_{I},e_{J})&=\frac{1}{\rho f'(\rho)}\left(u\dot{v}-\dot{u}v\right)\delta_{IJ}.
		\end{align}
		Hence, the umbilicity factor $\lambda$ satisfies
		\begin{align}\label{eq_lambda}
		\lambda&=\frac{1}{\rho f'(\rho)}\left(u\dot{v}-\dot{u}v\right).
		\end{align}
		By a straightforward computation using Proposition \ref{prop_curvatureextension}, we find that
		\begin{align}
		\nabla_{e_{0}}\eta&=\left(\ddot{v}-\dot{v}^{2}u\frac{f''(\rho)}{(f'(\rho))^2}\right)\partial_{v}-\left(\ddot{u}-\dot{u}^{2}v\frac{f''(\rho)}{(f'(\rho))^2}\right)\partial_{u}.
		\end{align}
		On the other hand, from the umbilicity of $P^n$, we know that
		\begin{align}
		\nabla_{e_{0}}\eta&=\lambda e_{0}
		\end{align}
		and thus
		\begin{align}
		\begin{split}\label{eq_secondorder}
		\lambda\dot{u}&=-\left(\ddot{u}-\dot{u}^{2}v\frac{f''(\rho)}{(f'(\rho))^2}\right),\\
		\lambda\dot{v}&=\phantom{-(}\;\ddot{v}-\dot{v}^{2}u\frac{f''(\rho)}{(f'(\rho))^2}.
		\end{split}
		\end{align}
		As $\mathfrak{h}(e_{0},e_{J})=0$, we conclude that umbilicity of $P^{n}$ with umbilicity factor $\lambda$ is indeed equivalent to~\eqref{eq_lambda}, \eqref{eq_secondorder}. Taking a derivative of \eqref{eq_lambda}, we see that
		\begin{align}\label{eq_derivative_lambda}
			\dot{\lambda}=\frac{1}{\rho f'(\rho)}\left(u\left(\ddot{v}-\lambda\dot{v}-\dot{v}^2u\frac{f''(\rho)}{(f'(\rho))^2}\right)-v\left(\ddot{u}+\lambda\dot{u}-\dot{u}^{2}v\frac{f''(\rho)}{(f'(\rho))^2}\right)\right).
		\end{align}
	We can therefore again verify that $\lambda$ is constant using the second order system \eqref{eq_secondorder}. Moreover, Equation \eqref{eq_lambda} and the parametrization by proper time \eqref{eq_arclength} imply that
	\begin{align}\label{eq_rhodot}
		(\dot{\rho})^2&=\frac{1}{(f'(\rho))^2}\left(\dot{v}u+\dot{u}v\right)^2
		=\rho^2\lambda^2-\frac{2uv}{(f'(\rho))^2F(\rho)}
		=\rho^2\lambda^2-h(\rho),
	\end{align}
	concluding the proof of the first claim.
	
	Let us now assume that $\gamma$ is a future directed timelike with $\dot{\rho}\not=0$ everywhere and everywhere orthogonal to $\mathcal{N}$. Assume further that $\gamma$ satisfies the first order system \eqref{symmetricODEs} for some constant $\lambda$. Using that $\dot{\rho}f'(\rho)=\dot{u}v+\dot{v}u$, the system \eqref{symmetricODEs} immediately implies that
	\[
		(2F(\rho)\dot{v}\dot{u}+1)h(\rho)=0
	\]
	along $\gamma$.
	Since $\dot{\rho}\not=0$ along $\gamma$, $h$ can only vanish for finitely many $s_i\in I$ (in fact $\gamma$ can cross each Killing horizon at most once as $\dot\rho\not=0$ has a fixed sign), so that $2F(\rho)\dot{v}\dot{u}+1=0$ and $u(s)\not=0\not=v(s)$ almost everywhere along $\gamma$. By continuity, $\gamma$ is parametrized by proper time everywhere. Using this and taking one radial derivative of \eqref{eq_ODE_central}, we see that
	\begin{align*}
		h'(\rho)=\frac{1}{K}\left(1-\frac{f''(\rho)f(\rho)}{(f'(\rho))^2}\right)
		=\frac{2}{f'(\rho)F(\rho)}\left(1-\frac{f''(\rho)f(\rho)}{(f'(\rho))^2}\right)
		=-\frac{4\dot{u}\dot{v}}{f'(\rho)}\left(1-\frac{f''(\rho)f(\rho)}{(f'(\rho))^2}\right)
	\end{align*}
	along $\gamma$. Taking a derivative of the second equation in \eqref{symmetricODEs} with respect to the curve parameter $s$ and using that $\lambda$ is constant by assumption, we see that
	\[
		2\dot{\rho}\ddot{\rho}=2\lambda^2\dot{\rho}{\rho}-h'(\rho)\dot{\rho}.
	\]
	Since $\dot{\rho}\not=0$ everywhere along $\gamma$ by assumption, we get
	\begin{align*}
		0&=\ddot\rho-\lambda^2\rho+\frac{h'(\rho)}{2}\\
		&=\frac{1}{f'(\rho)}\left(v\ddot{u}+u\ddot{v}+2\dot{u}\dot{v}-\frac{f''(\rho)}{(f'(\rho))^2}(\dot{u}v+u\dot{v})^2+\frac{f'(\rho)h'(\rho)}{2}-\lambda(\dot{v}u-v\dot{u})\right)\\
		&=\frac{1}{f'(\rho)}\left(u\left(\ddot{v}-\lambda\dot{v}-\dot{v}^2u\frac{f''(\rho)}{(f'(\rho))^2}\right)+v\left(\ddot{u}+\lambda\dot{u}-\dot{u}^2v\frac{f''(\rho)}{(f'(\rho))^2}\right)\right).
	\end{align*}
	Invoking again that $\lambda$ is constant and using the explicit form of its derivative \eqref{eq_derivative_lambda}, we can conclude that indeed the second order system \eqref{eq_secondorder} is satisfied along with equation \eqref{eq_lambda}. Therefore, $\gamma$ is the profile curve of a symmetric photon surface $P^n$ with umbilicity factor $\lambda$.
	
	Lastly, let us address the photon sphere case. Assume now that $P^n$ is a symmetric photon sphere, i.e., a symmetric photon surface with $\rho=\rho_*>0$ along $P^n$. Then the system of ODEs \eqref{symmetricODEs} is satisfied by the above analysis. Moreover, as $P^n$ is timelike by assumption, we know that $\rho_*\not=r_i$ for all $1\le i\le N$, and that $h(\rho_*)>0$. We may thus work in $(t,r)$-coordinates with $\partial_t$ timelike. We conclude
	\[
		\left(\frac{h}{r^2}\right)'(\rho_*)=0
	\]
	by the \emph{photon sphere condition} \cite[Theorem 3.5, (3.20)]{cedgal}. Conversely, assume that $\gamma$ is a future directed timelike curve satisfying \eqref{symmetricODEs} with $\rho\vert_{\gamma}=\rho_*$. From the second equation in \eqref{symmetricODEs}, we see that
	\[
		0=\rho_*^2\lambda^2-h(\rho_*),
	\]
	so that $h(\rho_*)\ge 0$ with equality if and only if $\lambda=0$. However, as $\dot u\not=0\not=\dot v$ along $\gamma$, $h(\rho_*)=0$ and $\lambda=0$ imply that $u=v=0$ along $\gamma$, where we used the first equation in \eqref{symmetricODEs} for $\lambda$. In particular, $\gamma$ is constant, which gives a contradiction. Hence, $h(\rho_*)>0$. Invoking again the results of Cederbaum--Galloway \cite[Theorem 3.5]{cedgal}, $P^n$ is a photon sphere. This concludes the proof.
\end{proof}

	Note that the converse claim in Proposition \ref{prop_charakterizationsymmetricphotonsurfaces} only addresses the cases when either $\dot\rho\not=0$ or $\dot\rho=0$ everywhere along the  profile curve. However, we see from the system of ODEs \eqref{symmetricODEs} that there can also be isolated parameter values $\overline{s}\in I$ of the profile curve $\gamma$ of a symmetric photon surface $P^n$ such that $\dot\rho(\overline{s})=0$.  The radii $\rho(\overline{s})$ for such parameter values $\overline{s}\in I$ depend only on the value of $\lambda^2$, in consistency with the second equation in \eqref{symmetricODEs}.
	This subtlety in the analysis of symmetric photon surfaces was studied and resolved by Cederbaum--Jahns--Vi\v{c}\'{a}nek-Mart\'{i}nez in \cite{cedoliviasophia}, showing that at any such point, a symmetric photon surface can be regularly joined to a reflection of itself across an appropriate $\{t=\operatorname{const.}\}$-slice, see Figure~\ref{fig:casesintro}. As
	\[
		\dot\tau=\frac{1}{h(\rho)f'(\rho)}(u\dot v-v\dot u)
	\]
	holds along the profile curve $\gamma$ whenever defined, we see that \eqref{symmetricODEs} is indeed equivalent to the system of ODEs derived by Cederbaum--Galloway (\cite[Lemma 3.4]{cedgal}). In \cite{cedoliviasophia}, Cederbaum--Jahns--Vi\v{c}\'{a}nek-Mart\'{i}nez completely analyze solutions to \eqref{symmetricODEs} globally in a domain of outer communication in $(\widetilde{M},\widetilde{g})$ (in $(t,r)$-coordinates). Their analysis of the behavior of solutions near parameter values $\overline{s}$ with $\dot\rho(\overline{s})=0$ is local in nature and hence applies in any region of $(\widetilde{M},\widetilde{g})$ where $h>0$. Note that by \eqref{symmetricODEs}, $\dot\rho(\overline{s})=0$ is excluded in regions of $(\widetilde{M},\widetilde{g})$ where $h<0$. Thus, to complete their analysis globally in $(\widetilde{M},\widetilde{g})$, it remains to discuss the properties of symmetric photon surfaces in regions where $h<0$ and to analyze the behavior of solutions of \eqref{symmetricODEs} across Killing horizons $\{r=r_i\}$, as these are the cases left open in \cite{cedoliviasophia}.
	
	As we know that $\lambda$ is constant along $P^n$ by Proposition \ref{prop_constantumbillicalfactor}, this gives us some a priori information how symmetric photon surfaces extend into the generalized Kruskal--Szekeres spacetime assuming that they indeed cross a Killing horizon. First, it is important to note that any choice of time orientation in a $h>0$-region of $(\widetilde{M},\widetilde{g})$ fixes the sign of the umbilicity factor $\lambda$ simultaneously for all symmetric photon surfaces in said region. Other than in the analysis in \cite{cedoliviasophia} where only a single such $h>0$-region  was considered, our choice of time orientation (see the end of Section \ref{sec_extension}) forces different signs on $\lambda$ in different copies of the said region. More explicitly, $\lambda>0$ holds in $h>0$-regions carrying a $+$ and $\lambda<0$ in $h>0$-regions carrying a $-$. Furthermore, note that there is no restriction on $\lambda$ in $h<0$-regions (no matter the choice of time orientation), and indeed all values of $\lambda$, including $\lambda=0$, do occur. For example, in $h<0$-regions the $\{t=\operatorname{const.}\}$-slices are the unique symmetric photon surfaces with $\lambda=0$, which cross the Killing horizon once through the bifurcation surface $\{u=v=0\}$. 
	
	Hence, a symmetric photon surface can never extend into two $h>0$-regions carrying opposite signs. Thus, any symmetric photon surface approaching a Killing horizon from a $h>0$-region can only cross said horizon away from the bifurcation surface, and only into a $h<0$-region.

	In what follows, we will concentrate our analysis on symmetric photon surfaces in regions carrying a $+$ as any symmetric photon surface in a region carrying a $-$ arises as the point reflection of a symmetric photon surface in the corresponding region carrying a $+$ (in any $(u,v)$-coordinate patch).
	\newpage
	\begin{thm}\label{prop_crossing}
		Let ${P}^n$ be a symmetric photon surface with umbilicity factor $\lambda$ in a generalized Kruskal--Szekeres spacetime $(\widetilde{M},\widetilde{g})$. Assume that all positive zeroes $r_1,\dotsc,r_N$ of $h$ are simple.
 		If $\rho\to r_i$ along ${P}^n$ for some $1\le i\le N$ then ${P}^n$ crosses the Killing horizon $\{r=r_i\}$ in $(\widetilde{M},\widetilde{g})$. In fact, it will cross the Killing horizon $\{r=r_i\}$ away from the bifurcation surface, unless $\lambda=0$. If, conversely, $\lambda=0$, it must cross the Killing horizon $\{r=r_i\}$ through the bifurcation surface.
	\end{thm}

	\begin{bem}\label{bem_crossing1}
		Before proving Theorem \ref{prop_crossing}, let us briefly mention for the convenience of the reader that the solution analysis in \cite{cedoliviasophia} distinguishes between different cases for fixed values of $\lambda^2>0$ in relation to the \emph{effective potential} $v_{\text{eff}}$ defined as
		\begin{align}\label{eq_effectivepotential}
		v_{\text{eff}}(r):=\frac{h(r)}{r^2}.
		\end{align}
		Any critical point of $v_{\text{eff}}$ (in a region where $h>0$) corresponds to a photon sphere, and if $P^n$ is not a photon sphere, then $\dot\rho$ vanishes at an isolated radius $r_\lambda$ along $P^n$ if and only if $v_{\text{eff}}(r_\lambda)=\lambda^2$. Away from horizons, i.e., where the time function $\tau=K\ln\btr{\frac{v}{u}}$ is well-defined, these facts generalize to our setting and the system of ODEs \eqref{symmetricODEs} in particular implies that
		\begin{align}\label{eq_carlasophiaoliviaODE}
		\frac{\d \rho}{\d \tau} = \mp h\sqrt{1-\lambda^{-2}v_{\text{eff}}}
		\end{align}
		for symmetric photon surfaces with $\lambda\not=0$. We refer to \cite[Section 3]{cedoliviasophia} for more details and enlightening figures. see also Figure \ref{fig_effpotential}.
		In particular, if $\frac{\d \rho}{\d \tau}\not=0$ on an open neighborhood of radii, then the profile curve can be written as the graph of a function $T_\lambda$ on said neighborhood, where 
		\begin{align}\label{eq_tlambda}
		T_\lambda(\rho)=\mp\int\limits^\rho\frac{1}{h(r)\sqrt{1-v_{\text{eff}}^\lambda(r)}}\d r,
		\end{align}
		with $v_{\text{eff}}^\lambda:=\lambda^{-2}v_{\text{eff}}$. Hence, if the open interval $(r_i,r_i{+1})$ corresponds to an $h<0$-region then $P^n$ can be globally written as a graph of $T_\lambda$ in this region, and approaches both $\{r=r_i\}$ and $\{r=r_{i+1}\}$. By Theorem \ref{prop_crossing} it will cross both of these horizons into (different) $h>0$-regions (unless $r_i=0$ or $r_{i+1}=\infty$, in which case $P^n$ will cross one horizon and approach the singularity $\rho=0$ or $\rho=\infty$, respectively).
	\end{bem}
	\begin{proof}
		As already discussed above, the case of $\lambda=0$ occurs only in $h<0$-regions and only for $\{t=\text{const.}\}$-slices. This forces $\lambda=0$-symmetric photon surfaces to extend through the bifurcation surface and stops them from crossing any Killing horizon away from its bifurcation surface. 
		
		Now, let us consider the case $\lambda\not=0$. We further assume without loss of generality that as $r\to r_i$, $P^n$ approaches the Killing horizon  from a region with $h>0$ corresponding to the first quadrant in the generalized Kruskal--Szekeres coordinates, i.e., $u,v>0$.\footnote{In particular, this addresses the possibly most interesting case of a symmetric photon surface in the domain of outer communication approaching a black hole or white hole region.} All other cases follow from almost identical arguments, possibly changing some signs and powers.

		Since $h\to 0$ as $r\to r_i>0$, we have $v_{\text{eff}}\to0$, so there exists $\delta>0$, such that ${\frac{\d \rho}{\d \tau}\not=0}$ on $(r_i,r_i+\delta)$ by \eqref{eq_carlasophiaoliviaODE}. Hence, in $\R\times (r_i,r_i+\delta)$, the radial profile can be written as the graph of a function $T_\lambda$ given by \eqref{eq_tlambda}.
		Hence $T_\lambda$ is the primitive of $\mp \frac{1}{h_\lambda}$, where ${h_\lambda\definedas h(r)\sqrt{1-v_{\text{eff}}^\lambda(r)}}$ with $h_\lambda(r_i)=0$ and ${h'_\lambda(r_i)=h'(r_i)\not=0}$. In particular, Proposition \ref{prop_central} guarantees the existence of a strictly increasing solution $f^\lambda_i$ of \eqref{eq_ODE_central} with respect to $h_\lambda$ and $K_i=\frac{1}{h'(r_i)}$ on $(r_i-\delta',r_i+\delta')$ for some appropriate $0<\delta'\le\delta$. Notice that $K_i\ln(f^\lambda_i)$ is a primitive of $\frac{1}{h_\lambda}$ on $(r_i,r_i+\delta')$, which yields that
		\begin{align}\label{eq_graphTathorizon}
		T_\lambda=\mp\left(K_i\ln(\btr{h})+\frac{K_i}{2}\ln(1-v_{\text{eff}}^\lambda)+R_{\lambda,i}\right)+C,
		\end{align}
		on $(r_i,r_i+\delta')$ by the fundamental theorem of calculus, where $R_{\lambda,i}$ is a well-defined, regular remainder function on $(r_i-\delta',r_i+\delta')$, cf.\ Proposition \ref{propcentralA} below, and where $C$ is a constant of integration. For simplicity, we will only address the $-$ case, as the $+$ case follows analogously. As $u,v>0$, the explicit expressions for the coordinate functions $u$, $v$ given in Remark \ref{bem_isometry} and for a solution of \eqref{eq_ODE_central} given in Remark \ref{bem_central} yield that
		\begin{align}
		\begin{split}\label{eq_coordinates}
			v(r)&=(1-v_{\text{eff}}^\lambda)^{-\frac{1}{4}} \exp\left(\frac{1}{2K_i}\left(R_i-R_{\lambda,i}\right)\right)\\
			u(r)&=h(1-v_{\text{eff}})^{\frac{1}{4}}\exp\left(\frac{1}{2K_i}\left(R_i+R_{\lambda,i}\right)\right),
		\end{split}
		\end{align}
		where we recall that $R_i:=\int\limits_{r_i}^r\frac{1-K_ih'}{h}$ is regular on $(r_i-\delta',r_i+\delta')$, cf. Proposition \ref{propcentralA} below.
		Thus, $u(r)\to 0$ and $v(r)$ converges to a strictly positive constant as $r\to r_i$. In particular, the symmetric photon surface does not go towards the bifurcation surface. As $h<0$ in $Q_2$, we note that a symmetric photon surface in Quadrant $Q_2=\{v>0,u<0\}$ with the same umbilicity factor $\lambda$ can similarly be described as the graph of $T_\lambda$, cf. Remark \ref{bem_crossing1}. Choosing the same constant of integration, one can verify that $u,v$ converge to the same values as $r\to r_i$, in fact $u,v$ still satisfy \eqref{eq_coordinates} on $(r_i-\delta',\delta)$. Thus the symmetric photon surface regularly extends across the Killing horizon (with its regularity depending on the regularity of $h$ in view of Theorem \ref{thm_main1}). This concludes the proof.
	\end{proof}
	\begin{bem}\label{bem_crossing2}
		Note that by \eqref{eq_carlasophiaoliviaODE} we can always locally write the profile curve as a graph over $r$ whenever $\dot\rho\not=0$, and otherwise employ the local result of Cederbaum--Jahns--Vi\v{c}\'{a}nek-Mart\'{i}nez \cite[Theorems 3.8, 3.9, 3.10]{cedoliviasophia}. Thus any maximally extended symmetric photon surface which crosses at least one Killing horizon either approaches the singularity $\rho=0$ or can be indefinitely extended in the generalized Kruskal--Szekeres extension. It is easy to see that the Kruskal--Szekeres extension of the Schwarzschild spacetime is an example where $\rho\to0$ for all symmetric photon surfaces crossing the horizon, see Figure \ref{fig_photonsurfaces}. On the other hand, the sub-extremal Reissner--Nordstr\"om spacetimes contain examples of indefinitely extended symmetric photon surfaces crossing infinitely many Killing horizons, see Figure \ref{fig_RNphotonsurfaces}.
		
		To see this, recall that the sub-extremal Reissner--Nordstr\"om spacetime with positive mass and non-trivial charge corresponds to the choice $h=1-\frac{2m}{r}+\frac{q^2}{r^2}$ with $m>\btr{q}>0$ (in spherical symmetry). Clearly, $h$ has two strictly positive, simple zeros $r_1$, $r_2$. Hence the generalized Kruskal--Szekeres spacetime, containing non-degenerate Killing horizons corresponding to the sets $\{r=r_1\}$, $\{r=r_2\}$, can be constructed. Note that on $(r_2,\infty)$, the corresponding effective potential $v_{\text{eff}}$ attains exactly one strict maximum with value $\lambda_*$. Following the construction and analysis of Cederbaum--Jahns--Vi\v{c}\'{a}nek-Mart\'{i}nez in \cite{cedoliviasophia}, we see that for $0<\lambda<\lambda_*$ there exists a symmetric photon surface in Region~\rom{1}$+$ that approaches the Killing horizon $\{r=r_2\}$ both for $t\to\infty$ and $t\to-\infty$. Hence, the photon surfaces crosses both the subsets $\{u=0\}$ and $\{v=0\}$ of the Killing horizon by Theorem \ref{prop_crossing}. As $h<0$ in any copy of Region~\rom{2}, the photon surface will further cross the Killing horizon $\{r=r_1\}$ into a copy of Region~\rom{3}$+$. Note that $h>0$ in Region~\rom{3} and $v_{\text{eff}}\to\infty$ as $\rho\to0$. Hence, there exists another turning point, i.e., a radius where $\dot\rho=0$ and where the surface is extended by reflection just as in the analysis by Cederbaum--Jahns--Vi\v{c}\'{a}nek-Mart\'{i}nez and thus approaches the Killing horizon $\{r=r_1\}$ again to pass into another copy of Region~\rom{2}. From there, it will again approach $\{r=r_2\}$ and extend into a copy of Region~\rom{1}$+$. Hence, it extends indefinitely. Again, see Figure \ref{fig_RNphotonsurfaces} below.
		
		Note that by the Penrose singularity theorem, $\{r=0\}$ still remains a causal singularity. Moreover, in the case of Reissner--Nordstr\"orm, $\{r=0\}$ is also a spacelike singularity by Proposition \ref{prop_kretschmann}.
	\end{bem}
	
	\begin{bem}\label{bem_crossing3}
		Note that the indefinitely extended symmetric photon surfaces in Reissner--Nordstr\"om discussed in Remark \ref{bem_crossing2} , see also Figure \ref{fig_RNphotonsurfaces}, are ``trapped'', between the singularity $\rho=0$ and the asymptotic end $\rho=\infty$. This provides an example questioning whether trapping of null geodesics should mean between a horizon and infinity or, more generally, between a singularity and $\infty$. See also \cite[Subsection 3.1]{cedoliviasophia} for a related example of a symmetric photon surface  trapped between the singularity and $\infty$ in superextremal Reissner--Nordstr\"om.
	\end{bem}
	\vfill\newpage
	\begin{figure}[H]
		\centering
		\includegraphics[scale=0.75]{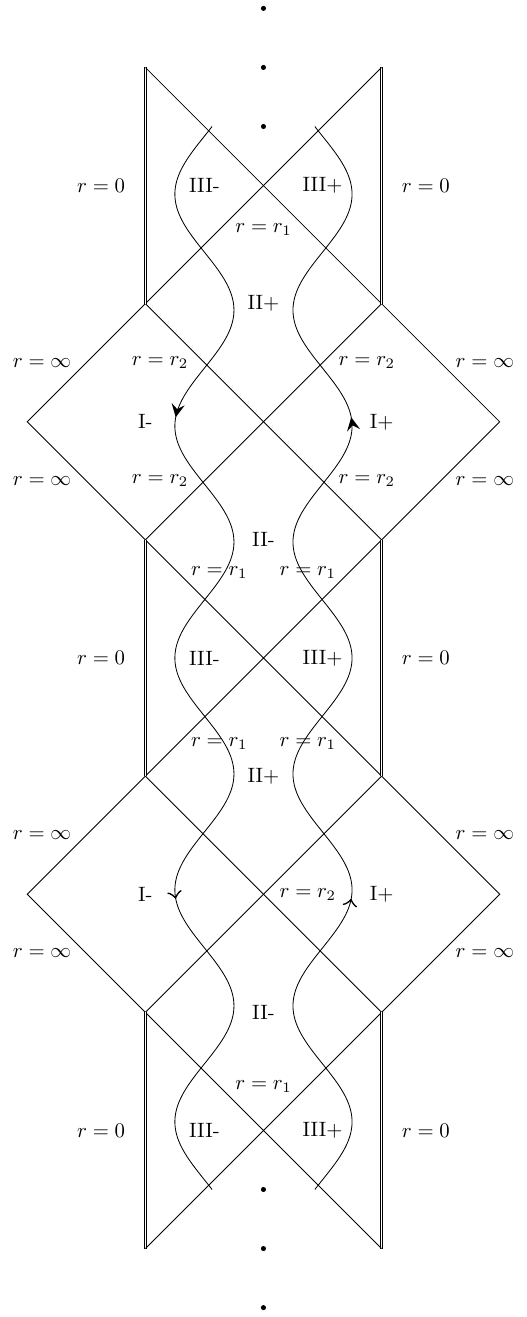}
		\caption{Indefinitely extendable photon surfaces in the generalized Kruskal--Szekeres spacetime of the subextremal Reissner--Nordstr\"om spacetime writing $r=\rho$ to match with the usual names in Reissner--Nordstr\"om.}
		\label{fig_RNphotonsurfaces}
	\end{figure}
	\newpage
	Last but not least, let us address the asymptotic behavior of symmetric photon surfaces as $\rho\to\infty$ if $h>0$ on $(r_N,\infty)$. For simplicity, we will assume a mild version of asymptotic flatness, requiring $h\to 1$ as $\rho\to\infty$ but without specifying any decay rate. Cederbaum--Jahns--Vi\v{c}\'{a}nek-Mart\'{i}nez \cite{cedoliviasophia} conjectured that any part of such a photon surface tending towards $\rho\to\infty$ should asymptote to the one-sheeted hyperboloid in the Minkowski spacetime and hence approach a lightcone in the given spacetime. Note that a concise notion of lightcones in a copy of Region~\rom{1}$+$ in generalized Kruskal--Szekeres coordiantes is conveniently given by the principal null hypersurfaces $\{v=\operatorname{const.}\}$, $\{u=\operatorname{const.}\}$ for any positive constant, respectively. Due to our choice of time-orientation, we call the sets $\{u=\text{const.}\}$ and $\{v=\text{const.}\}$ in Region~\rom{1}+ the future-directed lightcones and the past-directed lightcones, respectively. We prove the conjectured behavior with the next proposition utilizing the existence of generalized Kruskal--Szekeres coordinates. Note that sufficiently far out, $\dot\rho=0$ holds everywhere, and the second author found an explicit formula for the metric of a symmetric photon surface whenever $\dot\rho=0$, cf. \cite{wolff} Remark 3.5. In particular, this establishes that the metric converges to the metric of the one-sheeted hyperboloid with the precise rate of convergence depending on the asymptotic behavior of $h$.
	\begin{prop}\label{prop_lightconeasymptotics}
		Let ${P}^n$ be a symmetric photon surface in the domain of outer communication Region~\rom{1}+ of a generalized Kruskal--Szekeres spacetime under the same assumptions on $h$ as in Theorem \ref{prop_crossing}, and assume that $h\to1$ as $\rho\to\infty$. If $\rho\to\infty$ along some part of $P^n$ then this part of $P^n$ asymptotes to a lightcone.
	\end{prop}
	\begin{proof}
	Consider a part of $P^n$ with $\rho\to\infty$. As $h\to 1$, we note that by \eqref{symmetricODEs} $\dot\rho\not=0$ for $\rho$ large enough and we pick a point $(u_0,v_0)$ with radius $\rho_0\ge r_N$, such that $\dot\rho\not=0$ along $P^n$ for all $\rho\ge \rho_0$. Without loss of generality, $v_{\text{eff}}^\lambda<1$ for all $\rho\ge \rho_0$.
	Note that on $(r_N,\infty)$ we can express the solution $f_N$ of \eqref{eq_ODE_central} as
	\[
		f_N=\exp\left(\frac{1}{K_N}\int\frac{1}{h}+C\right),
	\]
	see Appendix \ref{app_solving} below.
	From this, we derive from Remark \ref{bem_isometry} that
	\begin{align*}
		v(r)&=v_0\exp\left(\frac{1}{2K_N}\left(\,\int\limits_{\rho_0}^r\frac{1}{h}+T_\lambda\right)\right)
		=v_0\exp\left(\frac{1}{2K_N}\left(\,\int\limits_{\rho_0}^r\frac{1}{h}\mp\frac{1}{h_\lambda}\right)\right)\\
		u(r)&=u_0\exp\left(\frac{1}{2K_N}\left(\,\int\limits_{\rho_0}^r\frac{1}{h}-T_\lambda\right)\right)
		=u_0\exp\left(\frac{1}{2K_N}\left(\,\int\limits_{\rho_0}^r\frac{1}{h}\pm\frac{1}{h_\lambda}\right)\right),
	\end{align*}
	where we used again that away from radii with $\dot\rho=0$, we can write the radial profile as a graph of a function $T_\lambda$ such that $T_\lambda$ satisfies \eqref{eq_graphTathorizon}, with $h_\lambda$ defined as above.
	We can therefore conclude that $v\to \text{const.}$ in the $-$ case, and $u\to\text{const.}$ in the $+$ case, respectively, once we show that the indefinite integral
	\[
	\int\limits_{\rho_0}^\infty\frac{1}{h(\rho)}\left(1-\frac{1}{\sqrt{1-v_{\text{eff}}^\lambda(\rho)}}\right)\d\rho
	\]
	converges. Since the integrand is strictly negative, it suffices to show that the integral remains bounded. Since $h\to 1$ as $\rho\to\infty$, there exists $\rho_1\ge\rho_0$, such that
	\[
		\sqrt{1-v_{\text{eff}}^\lambda}\left(1+\sqrt{1-v_{\text{eff}}^\lambda}\right)\ge 1
	\]
	for all $\rho\ge\rho_1$, recalling that $v_{\text{eff}}^\lambda(\rho)=\frac{h(\rho)}{\lambda^2\rho^2}$. Define 
	\[
		C_1\definedas-\int\limits_{\rho_0}^{\rho_1}\frac{1}{h(\rho)}\left(1-\frac{1}{\sqrt{1-v_{\text{eff}}^\lambda(\rho}}\right)\d\rho.
	\]
	Then, we estimate
	\begin{align*}
	0&\le -\int\limits_{\rho_0}^\infty\frac{1}{h(\rho)}\left(1-\frac{1}{\sqrt{1-v_{\text{eff}}^\lambda(\rho)}}\right)\d\rho\\
	&=C_1-\int\limits_{\rho_1}^\infty\frac{1}{h(\rho)}\left(1-\frac{1}{\sqrt{1-v_{\text{eff}}^\lambda(\rho)}}\right)\d\rho\\
	&=C_1+\int\limits_{\rho_1}^\infty\frac{1}{\sqrt{1-v_{\text{eff}}^\lambda(\rho)}\left(1+\sqrt{1-v_{\text{eff}}^\lambda(\rho)}\right)}\frac{1}{\lambda^2\rho^2}\d\rho\\
	&\le C_1+\frac{1}{\lambda^2}\int\limits_{\rho_1}^\infty\frac{1}{\rho^2}\d\rho< \infty
	\end{align*}
	This concludes the proof.
	\end{proof}
	We close this section with some figures illustrating our asymptotic ($\rho\to0,\,\rho\to\infty$) results, see also Remark \ref{bem_crossing2}. For simplicity, we assume $N=1$, $h\to 1$ as $\rho\to\infty$ and that $v_{\text{eff}}$ has exactly one strict, positive maximum in $(r_1,\infty)$ attaining the value $\lambda_*$. In particular, one can think of the Kruskal--Szekeres extension of the Schwarzschild spacetime. In analogy with Cederbaum--Jahns--Vi\v{c}\'{a}nek-Mart\'{i}nez, we distinguish between the three cases $0<\btr{\lambda}<\lambda_*$, $\btr{\lambda}=\lambda_*$ and $\btr{\lambda}>\lambda_*$.
	\begin{figure}[H]
		\centering
		\begin{subfigure}[b]{0.4\textwidth}
			\centering
			\includegraphics[width=\textwidth]{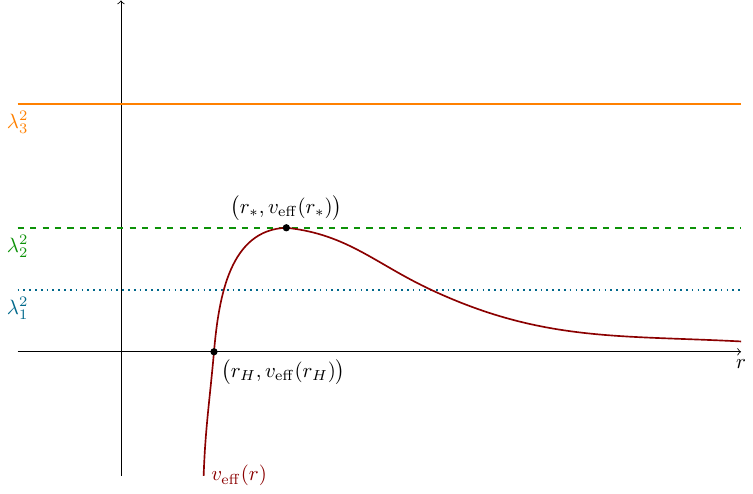}
			\caption{Effective potential with different values for $\btr{\lambda}$}
			\label{fig_effpotential}
		\end{subfigure}
		\hfill
		\begin{subfigure}[b]{0.4\textwidth}
			\centering
			\includegraphics[width=\textwidth]{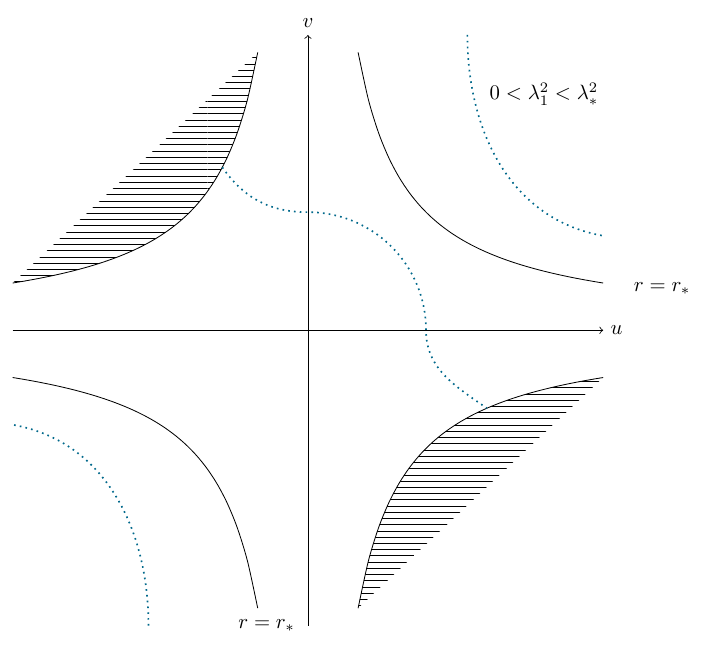}
			\caption{Photon surfaces with $\btr{\lambda}<\lambda_*$}
			\label{fig_kruskal1}
		\end{subfigure}\\
	\begin{subfigure}[b]{0.4\textwidth}
		\centering
		\includegraphics[width=\textwidth]{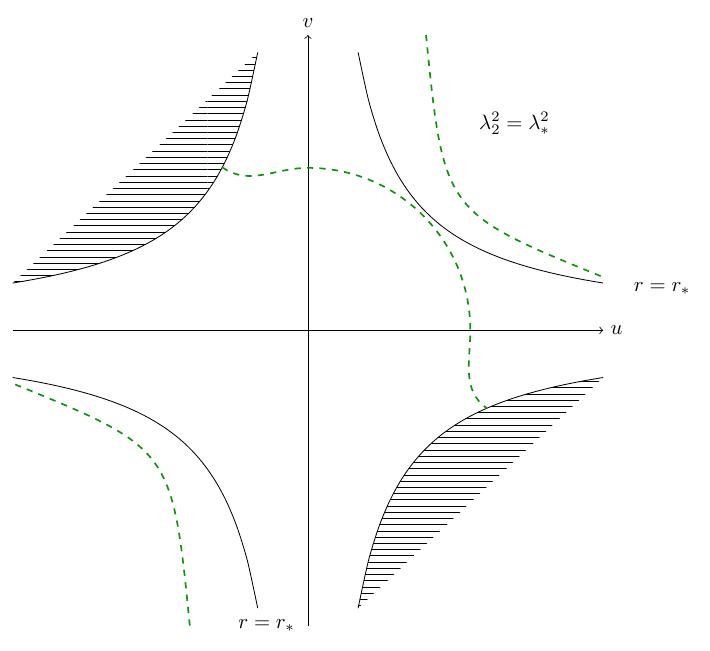}
		\caption{Photon surfaces with $\btr{\lambda}=\lambda_*$}
		\label{fig_kruskal2}
	\end{subfigure}
	\hfill
	\begin{subfigure}[b]{0.4\textwidth}
		\centering
		\includegraphics[width=\textwidth]{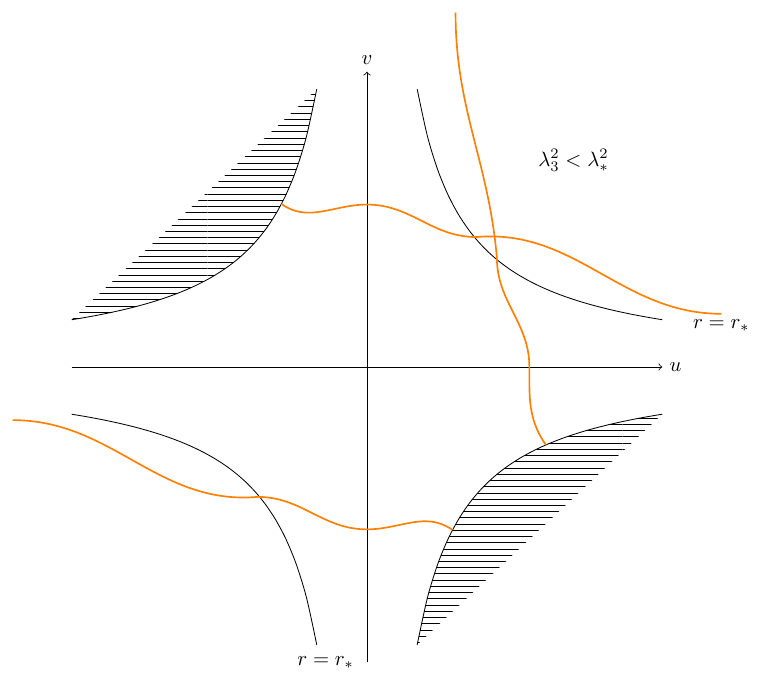}
		\caption{Photon surfaces with $\btr{\lambda}>\lambda_*$}
		\label{fig_kruskal3}
	\end{subfigure}
\caption{Symmetric photon surfaces in the Kruskal--Szekeres extension of the Schwarzschild spacetime.}
\label{fig_photonsurfaces}
	\end{figure}
	
\begin{appendix}
	\section{Solving the ODE}\label{app_solving}
		It remains to prove Proposition \ref{prop_central}. To motivate our analysis, note that, on $(r_i,r_{i+1})$ \eqref{eq_ODE_central} is equivalent to a linear ODE satisfying a local Lipschitz condition. In fact, on $(r_i,r_{i+1})$, \eqref{eq_ODE_central} is solved by
		\[
		f_i=\exp\left(\frac{1}{K_i}\int\frac{1}{h}+C_i \right)
		\]
		with any constant of integration $C_i$.
		Further note that the tortoise function $R^*=\int\frac{1}{h}$ considered in the approach of Brill--Hayward \cite{brillhayward} and in the approach of Schindler--Aguirre \cite{Schindler2018AlgorithmsFT} explicitly appears in the above form of $f$. As shown by them in their local analysis, we expect the tortoise function $R^*$ to be of the form
		\begin{align}\label{eqbrillhaywardanalogue}
		R^*=K_i\ln\btr{h}+R_i\text{ on }(r_i,r_{i+1}),
		\end{align}
		where $R_i\colon(r_{i-1},r_{i+1})\to\R$ is a smooth remainder function. However, we now want to expand \eqref{eqbrillhaywardanalogue} onto $(r_{i-1},r_{i+1})$ without taking an absolute value inside the logarithm. Instead, we differentiate \eqref{eqbrillhaywardanalogue} once and multiply by $h$. Thus, instead of \eqref{eqbrillhaywardanalogue} we consider the equation
		\begin{align}\label{eq_ODE_central2}
		1=K_ih'+R_i'h\text{ on }(r_i,r_{i+1})
		\end{align}
	 	on $(r_{i-1},r_{i+1})$. The following proposition shows that \eqref{eq_ODE_central2} is intricately related to \eqref{eq_ODE_central} and might indeed be a more favorable way of looking at Equation \eqref{eqbrillhaywardanalogue}.
		\begin{prop}\label{propcentralA}
			Let $(M,g)$ be a spacetime of class $\mathfrak{H}$. Then the following are equivalent:
			\begin{enumerate}
				\item[\emph{(i)}] There exists a strictly increasing smooth solution $f_i$ of \eqref{eq_ODE_central} on $(r_{i-1},r_{i+1})$.
				\item[\emph{(ii)}] There exists a smooth function $R_i$ on $(r_{i-1},r_{i+1})$ such that $h$ satisfies \eqref{eq_ODE_central2}, and \linebreak ${h'(r_i)=\frac{1}{2K_i}>0}$.
				\item[\emph{(iii)}] $h'(r_i)\not=0$.
			\end{enumerate}
		\end{prop}
		In particular, Proposition \ref{propcentralA} implies Proposition \ref{prop_central}. The uniqueness statement in Proposition \ref{prop_central} is easily derived from the proof of Proposition \ref{propcentralA}, see also the comments below. We first establish a higher order version of l'H\^opital's rule suited to our needs:
		\begin{lem}\label{lemmaauxilliary1}
			Let $f,g\colon(r_0-\varepsilon,r_0+\varepsilon)\to\R$ be $(k+1)$-times differentiable functions for some $r_0\in\R$, $\varepsilon>0$, where both functions have the unique zero $r_0$ in $(r_0-\varepsilon,r_0+\varepsilon)$, and $g>0$, $g'(r_0)\not=0$. Then
			\begin{align*}
			p\colon(r_0-\varepsilon,r_0+\varepsilon)\setminus\{r_0\}\to\R\colon r\mapsto\frac{f}{g},
			\end{align*}
			can be extended to $r_0$ in $C^{k}$.
		\end{lem}
		\begin{proof}
			We first prove the following claim: Let $v\colon(r_0-\varepsilon,r_0+\varepsilon)\to\mathbb{R}$ be a $(k+1)$-times differentiable function function satisfying $v(r_0)=0$. Then \begin{align*}
			u\colon(r_0-\varepsilon,r_0)\cup(r_0,r_0+\varepsilon)\to\mathbb{R}\colon r\mapsto\frac{v(r)}{r-r_0}
			\end{align*}
			is extendable to $r_0$ in $C^{k}$ with
			\begin{align*}
			u^{(l)}(r_0)&=\frac{v^{(l+1)}(r_0)}{n+1}
			\end{align*}
			for all $0\le l\le k$.
			
			To see this, we note that for all $0\le l\le k$
			\begin{align*}
			u^{(l)}(r)&=\frac{\sum_{j=0}^{n}(-1)^{l+j}\frac{l!}{j!}v^{(j)}(r)(q-r_0)^{j}}{(r-r_0)^{l+1}},
			\end{align*}
			away from $r_0$, a fact that can be proven by induction.
			By l'H\^{o}pital's rule, we find
			\begin{align*}
			\lim_{r\to r_0}u^{(l)}(r)&=\lim_{r\to r_i}\frac{\sum_{j=0}^{l}(-1)^{l+j}\frac{l!}{j!}v^{(l)}(r)(r-r_0)^{j}}{(r-r_0)^{l+1}}\\
			&=\lim_{r\to r_0}\frac{\sum_{j=0}^{l}(-1)^{l+j}\frac{l!}{j!}\left(v^{(j+1)}(r)(r-r_0)^{j}+j\,v^{(j)}(r)(r-r_i)^{j-1}\right)}{(n+1)(r-r_0)^{l}}\\
			&=\lim_{r\to r_0}\frac{v^{(l+1)}(r)}{l+1}\\
			&=\frac{v^{(l+1)}(r_0)}{l+1},
			\end{align*} 
			so $u$ is extendable in $C^{k}$ to $r_0$.
			
			We define the auxiliary functions
			\begin{align*}
			f_1&\colon(r_0-\varepsilon,r_0)\cup(r_0,r_0+\varepsilon)\to\mathbb{R}, r\mapsto \frac{g(r)}{r-r_0},\\
			f_2&\colon(r_0-\varepsilon,r_0)\cup(r_0,r_0+\varepsilon)\to\mathbb{R}, r\mapsto \frac{f(r)}{r-r_0}
			\end{align*}
			which are $(k+1)$-differentiable away from $r_0$. By the above, $f_1,f_2$ are extendable in $C^{k}$ and in particular $f_1(r_0)=g'(r_0)\not=0$, so $p=\frac{f_2}{f_1}$ is extendable in $C^{k}$ to $r_0$.
		\end{proof}
		With this in mind, we return to proving Proposition \ref{propcentralA}
		\begin{proof}[Proof of Proposition \ref{propcentralA}]
			We will prove the equivalences of (i),(ii) and (ii),(iii) separately.
			\begin{itemize}
				\item [(i)$\Rightarrow$(ii)] Let $f_i$ be a strictly increasing, smooth solution of \eqref{eq_ODE_central} on $(r_{i-1},r_{i+1})$. Taking a derivative, this implies	
				\begin{align*}
				1=K_ih'+R_i'h,
				\end{align*}
				with $R_i:=K_i\ln(f_i')$ on $(r_{i-1},r_{i+1})$, and hence $h'(r_i)=\frac{1}{K_i}\not=0$.
				\item [(ii)$\Rightarrow$(i)] Let $h$ satisfy \eqref{eq_ODE_central2} on $(r_{i-1},r_{i+1})$ for some smooth function $R_i$ on $(r_{i-1},r_{i+1})$. Define
				\[
				f_i:=h\exp\left(\frac{R_i}{K_i}\right),
				\]
				Then $f_i'>0$ and one can directly verify that $f_i$ satisfies \eqref{eq_ODE_central} on $(r_{i-1},r_{i+1})$.
			\end{itemize}
			\begin{itemize}
				\item [(ii)$\Rightarrow$(iii)] Let $h$ satisfy \eqref{eq_ODE_central2} on $(r_{i-1},r_{i+1})$ for some smooth function $R_i$ on $(r_{i-1},r_{i+1})$. In particular,
				\[
				1=K_ih'(r_i)+R_i'(r_i)h(r_i)=K_ih'(r_i),
				\]
				so $h'(r_i)=\frac{1}{K_i}\not=0$.
				\item [(iii)$\Rightarrow$(ii)] Let $h'(r_i)\not=0$. Define $K_i\definedas\frac{1}{h'(r_i)}$. By Lemma \ref{lemmaauxilliary1}, $\frac{1-K_ih'}{h}$ is smoothly extendable onto $(r_{i-1},r_{i+1})$. Then $h$ satisfies \eqref{eq_ODE_central2} on $(r_{i-1},r_{i+1})$ for $R_i:=\int\frac{1-K_ih'}{h}$, and ${h'(r_i)=\frac{1}{K_i}}$ holds by definition.
			\end{itemize}
		\end{proof}
		Note that the constructions of $R_i$ respectively $f_i$ in (i)$\Rightarrow$(ii) and (i)$\Leftarrow$(ii) are inverse to one another. In particular, a strictly increasing solution $f_i$ of \eqref{eq_ODE_central} is fully determined by the choice of the smooth remainder function $R_i$, and vice versa. However, $R_i$ is uniquely determined by \eqref{eq_ODE_central2} up to a constant, so all strictly increasing solutions $f_i$ of \eqref{eq_ODE_central} are indeed uniquely determined up to scaling. Hence, there is a unique strictly increasing solution $f_i$ of \eqref{eq_ODE_central} on $(r_{i-1},r_{i+1})$ such that $f_i'(r_i)=1$.
	\section{Curvature components}\label{app_curvature}
	For convenience, let us collect the relevant Christoffel symbols and the Ricci and scalar curvature of the generalized Kruskal-Szekeres extension.
	\begin{prop}\label{prop_curvatureextension}
		Let $g=(F\circ\rho)(\d u\d v+\d v\d u)+r^2g_\mathcal{N}$ with $F$ as in Definition \ref{extdefi1}. Then the relevant Christoffel symbols of $g$ in coordinates $(u,v,x^I)$ are the following:
		\begin{align*}
		\Gamma_{vv}^v&=-\frac{u}{f'^2}f'',\\
		\Gamma_{uu}^u&=-\frac{v}{f' 2}f'',\\
		\Gamma_{uI}^K&=\frac{v}{\rho f'}\delta_I^K,\\
		\Gamma_{vI}^K&=\frac{u}{\rho f'}\delta_I^K,\\
		\Gamma_{IJ}^u&=-\frac{u\rho}{2K}\left(g_{\mathcal{N}}\right)_{IJ},\\
		\Gamma_{IJ}^v&=-\frac{v\rho}{2K}\left(g_{\mathcal{N}}\right)_{IJ},\\
		\Gamma_{IJ}^K&=\left(\Gamma^N\right)_{IJ}^K,
		\end{align*}
		where $\left(\Gamma^N\right)_{IJ}^K$ denote the Christoffel symbols on $(\mathcal{N},g_\mathcal{N})$ respectively. Moreover
		\begin{align*}
			\Ric_{uv}&=-\frac{K}{f'}\left(h''+\frac{(n-1)}{\rho}h'\right),\\
			\Ric_{IJ}&=\left(\Ric_{\mathcal{N}}\right)_{IJ}-\left((n-2)h+\rho h'\right)\left(g_{\mathcal{N}}\right)_{IJ},
		\end{align*}
		where $\Ric_\mathcal{N}$ denotes the Ricci curvature on $(\mathcal{N},g_\mathcal{N})$.
	All other Christoffel symbols and Ricci curvature components vanish.
	\end{prop}
	\begin{proof}
		Straightforward computation for warped products using \eqref{eq_ODE_central}.
	\end{proof}\nocite{*}
\end{appendix}

\bibliographystyle{plain}
\bibliography{biblio_kruskal}

\end{document}